\documentclass[a4paper,UKenglish,cleveref,thm-restate,notab]{lipics-v2021}
\nolinenumbers
\hideLIPIcs
\usepackage[nocompress]{cite}

\usepackage{mathtools}
\usepackage[overload]{empheq}
\usepackage[ruled, noend]{algorithm2e}
\usepackage{algpseudocode}

\newcommand{\etal}{\textit{et~al.\@}\xspace}

\graphicspath{{./figures/}}

\title{Faster, Deterministic and Space Efficient Subtrajectory Clustering}

\author
{Ivor van der Hoog}
{Department of Applied Mathematics and Computer Science, Technical University of Denmark, Denmark}
{vanderhoog@gmail.com}
{https://orcid.org/0009-0006-2624-0231}
{This project has received funding from the European Union's Horizon 2020 research and innovation programme under the Marie Sk\l{}odowska-Curie grant agreement No 899987, and from the Carlsberg Foundation via Eva Rotenberg's Young Researcher Fellowship CF21-0302 ``Graph Algorithms with Geometric Applications''.}

\author
{Thijs van der Horst}
{Department of Information and Computing Sciences, Utrecht University, the Netherlands \and  Department of Mathematics and Computer Science, TU Eindhoven, the Netherlands} {t.w.j.vanderhorst@uu.nl}
{https://orcid.org/0009-0002-6987-4489}
{}

\author
{Tim Ophelders}
{Department of Information and Computing Sciences, Utrecht University, the Netherlands \and 
Department of Mathematics and Computer Science, TU Eindhoven, the Netherlands}
{t.a.e.ophelders@uu.nl}
{https://orcid.org/0000-0002-9570-024X}
{Partially supported by the Dutch Research Council (NWO) under the project number VI.Veni.212.260.}

\authorrunning{I. van der Hoog, T. van der Horst, and T. Ophelders}
\Copyright{Ivor van der Hoog, Thijs van der Horst, and Tim Ophelders}

\acknowledgements{We thank Jacobus Conradi for pointing out an error in a preprint of this paper.}



\EventEditors{Keren Censor-Hillel, Fabrizio Grandoni, Joel Ouaknine, and Gabriele Puppis}
\EventNoEds{4}
\EventLongTitle{52nd International Colloquium on Automata, Languages, and Programming (ICALP 2025)}
\EventShortTitle{ICALP 2025}
\EventAcronym{ICALP}
\EventYear{2025}
\EventDate{July 8--11, 2025}
\EventLocation{Aarhus, Denmark}
\EventLogo{}
\SeriesVolume{334}
\ArticleNo{110}

\newcommand{\f}{Fr\'echet\xspace}

\newcommand{\dF}{\ensuremath{d_F}}
\mathchardef\mhyphen="2D
\newcommand{\FSD}[1][\Delta'] {\ensuremath{#1}\mhyphen\mathrm{FSD}}

\newcommand{\eps}{\varepsilon}
\newcommand{\from}{\colon}
\newcommand{\bigO}{\mathcal{O}}
\newcommand{\Ot}{\ensuremath{\tilde\bigO}}

\newcommand{\R}{\mathbb{R}}
\newcommand{\N}{\mathbb{N}}
\newcommand{\B}{\mathbb{B}}

\newcommand{\I}{\mathcal{I}}

\newcommand{\U}{\mathcal{U}}

\newcommand{\ver}{\mathrm{ver}}
\newcommand{\pre}{\mathrm{pre}}
\newcommand{\suf}{\mathrm{suf}}
\newcommand{\sub}{\mathrm{sub}}
\newcommand{\rev}{\overleftarrow}

\newcommand{\Cov}{\ensuremath{\mathrm{Cov}}}

\newcommand{\floor}[1]{\lfloor#1\rfloor}
\newcommand{\ceil}[1]{\lceil#1\rceil}

\newcommand{\SW}{\mathit{SW}}

\begin{document}

\maketitle

\begin{abstract}
    Given a trajectory $T$ and a distance $\Delta$, we wish to find a set $C$ of curves of complexity at most $\ell$, such that we can cover $T$ with subcurves that each are within Fr\'echet distance $\Delta$ to at least one curve in $C$.
    We call $C$ an $(\ell,\Delta)$-clustering and aim to find an $(\ell,\Delta)$-clustering of minimum cardinality.
    This problem variant was introduced by Akitaya~\textit{et al.}~(2021) and shown to be NP-complete.
    The main focus has therefore been on bicriteria approximation algorithms, allowing for the clustering to be an $(\ell, \Theta(\Delta))$-clustering of roughly optimal size.
    
    We present algorithms that construct $(\ell,4\Delta)$-clusterings of $\mathcal{O}(k \log n)$ size, where $k$ is the size of the optimal $(\ell, \Delta)$-clustering.
    We use $\mathcal{O}(n^3)$ space and $\mathcal{O}(k n^3 \log^4 n)$ time.
    Our algorithms significantly improve upon the clustering quality (improving the approximation factor in $\Delta$) and size (whenever $\ell \in \Omega(\log n / \log k)$).
    We offer deterministic running times improving known expected bounds by a factor near-linear in $\ell$.
    Additionally, we match the space usage of prior work, and improve it substantially, by a factor super-linear in $n\ell$, when compared to deterministic results.
\end{abstract}

\section{Introduction}
\label{sec:introduction}

    In subtrajectory clustering, the goal is to partition an input trajectory $T$ with $n$ vertices into subtrajectories and group them into \emph{clusters} such that all subtrajectories within a cluster have low \f distance to one another. 
    Clustering under the \f distance is a natural application of the \f distance and a well-studied topic~\cite{driemel2016clustering, cheng2023curve, buchin2022coresets, buchin2019approximating, gudmundsson2022cubic} with applications in, for example, map reconstruction~\cite{buchin2017clustering, buchin_improved_2020}. 
    In recent years, several variants of this algorithmic problem have been proposed~\cite{buchin_detecting_2011, agarwal_subtrajectory_2018, bruning_subtrajectory_2023, bruning_faster_2022}. Regardless of the variant, the subtrajectory clustering problem has been shown to be NP-complete~\cite{buchin_detecting_2011, agarwal_subtrajectory_2018, bruning_subtrajectory_2023}.

    We focus on the problem variant proposed by  Akitaya, Br\"{u}ning, Chambers, and Driemel~\cite{bruning_subtrajectory_2023}. 
     Given a trajectory $T$ and a distance $\Delta$, and some $\ell$, they compute what we call an \emph{$(\ell, \Delta)$-clustering} $C$ of $T$.
    Each cluster $Z \in C$ is a set of subtrajectories together with a center curve (the `reference curve' $P_Z$) of complexity at most $\ell$. Each curve in a cluster must have Fr\'echet distance at most $\Delta$ to the center and each point on $T$ must be present in at least one cluster. 
    The goal is to compute an $(\ell, \Delta)$-clustering of minimum cardinality. 
    Note that the parameter $\ell$ is necessary to not trivialize the problem. Indeed, if $P_Z$ may have arbitrary complexity, then a trivial $(n, 0)$-clustering exists consisting of a single cluster $Z$ where $Z = \{ T \}$.
    
    Akitaya~\etal~\cite{bruning_subtrajectory_2023} propose a bicriteria approximation scheme:
    Given $\ell$ and $\Delta$, let $k$ be the minimum size of an $(\ell, \Delta)$-clustering of $T$. The goal is  to compute an $(\ell, \Theta(\Delta))$-clustering of size $\bigO(f(k))$. 
    This paradigm was studied in~\cite{bruning_subtrajectory_2023, bruning_faster_2022, conradi2023finding} and previous results are summarised in Table~\ref{tab:results}.      
    \cite{bruning_subtrajectory_2023} computes an 
    $(\ell, 19\Delta)$-clustering of $\bigO(k \ell^2 \log( k \ell))$ size. 
    The running time and space bounds depend on the \emph{arc length} of $T$ relative to $\Delta$.
    Br\"{u}ning, Conradi and Driemel\cite{bruning_faster_2022} compute an $(\ell, 11\Delta)$-clustering of $\bigO(k \ell \log k)$ size (where the hidden constant is exceptionally large). 
    Their algorithm uses $\Ot(n^3)$ space and has $\Ot(kn^3)$ expected running time.
    Recently, Conradi and Driemel~\cite{conradi2023finding} improve both the size and the quality of the clustering.
    They compute an $(\ell, 11 \Delta)$-clustering of $\bigO(k \log n)$ size in $\Ot( n^4 \ell)$ space and $\Ot(k n^4 \ell + n^4 \ell^2)$ time.

    \begin{table}[t]
    \centering
    \begin{tabular}{l|l|l|l|l}
      {\# Clusters} &     $\Delta' =$ & {Time}                                   & {Space}                 & {Source}\\
    \hline\rule{0pt}{2.6ex}%
           \color{red}{$\bigO(k \ell^2 \log (k\ell))$} & \color{red}{$19\Delta$} & \color{red}{$\Ot(k \ell^4 \lambda^2 + n \lambda)$} & \color{red}{$\bigO(n + \lambda)$} & \color{red}{\cite{bruning_subtrajectory_2023}}\\
            \color{red}{$\bigO(k \ell \log k)$} & \color{red}{$11\Delta$} & \color{red}{$\Ot(k n^3 \ell)$}      & \color{red}{$\Ot(n^3)$}      & \color{red}{\cite{bruning_faster_2022}}\\
            $\bigO(k \log n)$           &      $11 \Delta$ & $\Ot(k n^4 \ell + n^4 \ell^2)$ & $\Ot(n^4 \ell)$ & \cite{conradi2023finding}\\
            $\bigO(k \log n )$            &      $4 \Delta$ & $\bigO(k n^3 \log^4 n)$                  & $\bigO(n^3)$  & Thm.~\ref{thm:yes-no-no}\\
    \end{tabular}
    \caption{
        Prior work and our result.
        The first two (red) rows indicate randomized results. 
        $k$ denotes the smallest $(\ell, \Delta)$-clustering size of $T$.
        $\lambda$ denotes the arc length of $T$ relative to $\Delta$.
    }
    \label{tab:results}
    \end{table}

\subparagraph*{Results.}
    We present a bicriteria approximation algorithm that uses $\bigO( k n^3 \log^4 n)$ time and $\bigO(n^3)$ space, and computes an $(\ell, 4\Delta)$-clustering of size $\bigO( k \log n)$. 
    When compared to previous works~\cite{bruning_subtrajectory_2023, conradi2023finding, bruning_faster_2022} our results:
    \begin{itemize}
        \item obtain deterministic results and improve the running time by a factor near-linear in $\ell$,
        \item match the space usage,
        \item improve the approximation in $\Delta$ from a factor $11$ to $4$, 
        \item asymptotically match the clustering size (whenever $\ell \in \Omega(\log n / \log k)$).
    \end{itemize}
    Compared exclusively to deterministic results~\cite{conradi2023finding}, we instead improve time by a factor near-linear in $n \ell$, space by a factor super-linear in $n \ell$, and obtain asymptotically equal clustering size for all $\ell$ (see also~\cref{tab:results}).

\subparagraph*{Methodology and contribution.}
    Our algorithm constructs a clustering iteratively by greedily adding a cluster that covers an approximately-maximum set of uncovered points on $T$.
    The challenge is to compute such a cluster.
    Previous work~\cite{bruning_faster_2022, bruning_subtrajectory_2023} presented randomized algorithms for constructing a cluster based on $\eps$-net sampling over the set of all candidate clusters.
    They shatter the set of candidate clusters and show that it has bounded VC dimension, which leads to their asymptotic approximation of $k$ --- the minimum size of an $(\ell, \Delta)$-clustering.
    The algorithm of Conradi and Driemel~\cite{conradi2023finding} is more similar to ours.
    They also simplify the input and iteratively select the cluster with the (exact) maximum coverage to obtain an $(\ell, \Delta)$-clustering of size $\bigO(k \log n)$.
    The key difference lies in finding the next cluster.
    Conradi and Driemel~\cite{conradi2023finding} explicitly consider a set of $\bigO(n^3 \ell)$ candidate clusters, which requires $\bigO(n^4 \ell)$ time and space to construct.
    
    We make two key contributions that distinguish us from prior works:
    First, we present a novel simplification algorithm that computes a curve $S$ such that we may restrict potential reference curves of clusters to be subcurves of $S$.
    This new curve simplification technique allows us to create a clustering where clusters have radius at most $4\Delta$ as opposed to $11\Delta$. 
    Second, we prove that we may restrict the reference curves to be one of two types:
    \begin{itemize}
        \item vertex-subcurves of $S$, which are subcurves that start and end at a vertex of $S$, \\ (we may furthermore only consider subcurves whose complexity is a power of $2$)
        \item and subedges of $S$, which are subcurves that are a subsegment of a single edge of $S$.
    \end{itemize}
    We prove that a greedy algorithm that exclusively adds maximal clusters where the reference curve is of one of these two types creates a clustering of size $\bigO(k \log n)$. 
    This characterization reduces the set of candidate clusters from $\Ot(n^3 \ell)$ to $\Ot(n^2)$ which significantly reduces the time spent compared to~\cite{conradi2023finding}.
    The geometric characterization of these subcurves allow us to compute candidate clusters on the fly, significantly reducing space usage. 

\section{Preliminaries}

    A \emph{(polygonal) curve} with $n$ vertices is a piecewise-linear map $P\from[1,n]\to\R^d$ whose breakpoints (called \emph{vertices}) are at each integer parameter, and whose pieces are called \emph{edges}.
    We denote by $P[a, b]$ the subcurve of $P$ that starts at $P(a)$ and ends at $P(b)$.
    If $a$ and $b$ are integers, we call $P[a,b]$ a \emph{vertex subcurve} of $P$.
    Let $|P|$ denote the number of vertices of~$P$.

\subparagraph*{\f distance.}
    A \emph{reparameterization} of $[1, n]$ is a non-decreasing surjection $f \from [0, 1] \to [1, n]$.
    Two reparameterizations $f$ and $g$ of $[1, m]$ and $[1, n]$, respectively, describe a \emph{matching} $(f, g)$ between two curves $P$ and $Q$ with $n$ and $m$ vertices, where for any $t \in [0, 1]$, point $P(f(t))$ is matched to $Q(g(t))$.
    A matching $(f, g)$ is said to have \emph{cost} $
        \max_t~\lVert P(f(t)) - Q(g(t)) \rVert,$  where $\lVert \cdot \rVert$ denotes the Euclidean norm.
    A matching with cost at most $\Delta$ is called a \emph{$\Delta$-matching}.
    The (continuous) \emph{\f distance} $\dF(P, Q)$ between $P$ and $Q$ is the minimum cost over all matchings.
    
\subparagraph*{Free space diagram.}
    The \emph{parameter space} of curves $P$ and $Q$ with $m$ and $n$ vertices, respectively, is given by the orthogonal rectangle $[1, m] \times [1, n]$.
    This parameter space is associated with a regular grid whose cells are the squares $[i, i+1] \times [j, j+1]$ for integers $i$ and~$j$.
    A point $(x, y)$ in the parameter space corresponds to the pair of points $P(x)$ and $Q(y)$.
    We say that $(x, y)$ is \emph{$\Delta$-free} if $\lVert P(x) - Q(y) \rVert \leq \Delta$.
    The \emph{$\Delta$-free space diagram} $\FSD[\Delta](P, Q)$ of $P$ and $Q$ is the set of $\Delta$-free points in the parameter space of $P$ and $Q$.
    The \emph{obstacles} of $\FSD[\Delta](P, Q)$ are the connected components of $([1, m] \times [1, n]) \setminus \FSD[\Delta](P, Q)$.

    Alt and Godau~\cite{alt95continuous_frechet} observe that the \f distance between $P[x_1, x_2]$ and $Q[y_1, y_2]$ is at most $\Delta$ if and only if there is a bimonotone path in $\FSD[\Delta](P, Q)$ from $(x_1, y_1)$ to $(x_2, y_2)$ (and $x_1 \leq x_2$ and $y_1 \leq y_2$).
    
\subparagraph*{Input and output.}
    Our input is a curve $T$ with $n$ vertices, which we will call the \emph{trajectory}, some integer parameter $\ell \geq 2$,
    and some distance parameter $\Delta \geq 0$.
    We consider \emph{clustering} subtrajectories of $T$ using \emph{pathlets}:
    
    \begin{definition}[Pathlet]
        An $(\ell, \Delta)$-\emph{pathlet} is a tuple $(P, \I)$ where $P$ is a curve with $|P| \leq \ell$ and $\I$ is a set of intervals in $[1, n]$, where $\dF(P, T[a, b]) \leq \Delta$ for all $[a, b] \in \I$.
        We call $P$ the \emph{reference curve} of $(P, \I)$.
\end{definition}

    \begin{figure}[b]
        \centering
        \includegraphics[page=1, width = 0.8\linewidth]{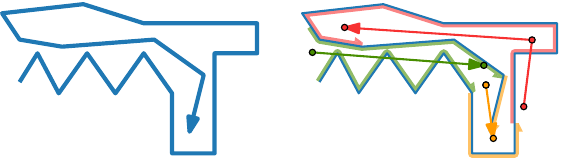}
        \caption{The trajectory $T$ (blue, left) is covered by three pathlets.
        Each pathlet is defined by a reference curve (green, red, yellow) and the subcurve(s) of $T$ the curve covers.
        }
        \label{fig:subtrajectory_clustering}
    \end{figure}

   We can see a pathlet $(P, \I)$ as a cluster, where the center is $P$ and all subtrajectories induced by $\I$ get mapped to $P$.
    See~\cref{fig:subtrajectory_clustering}.
    An $(\ell, \Delta)$-clustering of $T$ is defined as follows:

    \begin{definition}
        \label{def:cluster}
        An \emph{$(\ell, \Delta)$-clustering} $C$ is a set of $(\ell, \Delta)$-pathlets with $\bigcup\limits_{(P, \I) \in C} \bigcup\limits_{I \in \I} I=  [1, n]$.
    \end{definition}

    \noindent
    Throughout this paper, we let $k_\ell(\Delta)$ denote the smallest integer for which there exists an $(\ell, \Delta)$-clustering of size $k_\ell(\Delta)$.
    The goal is to find an $(\ell, \Delta')$-clustering $C$ where $|C|$ is not too large compared to $k_\ell(\Delta)$, and $\Delta' \in \bigO(\Delta)$.

\subparagraph{Weighting a cluster.}
We define a \emph{universe} $\U$ as any set of interior-disjoint closed intervals that together cover $[1, n]$. 
Given a fixed universe $\U$, we can weigh each pathlet by what we call its \emph{coverage}:

\begin{definition}
    The \emph{coverage} over $\U$ of a pathlet $(P, \I)$ is $\Cov_\U(P, \I) = \{I \in \U \mid I \subseteq \Cov(P, \I)\}$.
    The coverage of a set $C$ of pathlets is $\Cov_\U(C) = \sum_{(P, \I) \in C} \Cov_\U(P, \I)$.
\end{definition}

\noindent
Whenever $\U$ is clear from context we omit the subscript $\U$.

    \begin{definition}[Reference optimal]
         Let the universe $\U$ be fixed and let $C$ be a set of $(\ell, \Delta)$-pathlets. An $(\ell, \Delta)$-pathlet $(P, \I)$ is \emph{reference-optimal} if its coverage over $\U \setminus \Cov(C)$, i.e., $|\Cov(P, \I) \backslash \Cov(C)|$, is maximum over all $(\ell, \Delta)$-pathlets with the same reference curve.
\end{definition}

\begin{definition}
    Let the universe $\U$ be fixed and let $C$ be a set of $(\ell, \Delta)$-pathlets.
        An $(\ell, \Delta)$-pathlet $(P, \I)$ is \emph{optimal} whenever $|\Cov(P, \I) \backslash \Cov(C)|$ is maximum over all $(\ell, \Delta)$-pathlets.
    \end{definition}

    \noindent

\section{Algorithmic outline}
\label{sec:outline}

    Our algorithmic input is a trajectory $T$, an integer $\ell \geq 2$, and value $\Delta \geq 0$.
    We provide a high-level overview of our algorithm here.   
    Our approach can be decomposed as follows:
    
    \begin{enumerate}
        \item Reference curves may lie anywhere in the ambient space.
        Our first step is to restrict where these reference curves may lie.
        In~\cref{sec:simplification} we construct a $2\Delta$-simplification $S$ of $T$, and prove that for any $(\ell, \Delta)$-pathlet $(P, \I)$, there exists a subcurve $S[a, d]$ of $S$ for which $(S[a, d], \I)$ is an $(\ell+2-|\N \cap \{a, d\}|, \Delta')$-pathlet, where $\Delta' = 4\Delta$.
        Hence we may restrict our attention to pathlets where the reference curve is a subcurve of $S$, if we allow for a slightly higher complexity.
        This higher complexity is circumvented later on, to still give an $(\ell, \Delta')$-clustering.
        \item  In~\cref{sec:set_cover}, Given $S$ and $T$, we smartly create some universe $\U$. We prove, by adapting the argument for greedy set cover, that any algorithm that iteratively computes an optimal $(\ell, \Delta)$-pathlet outputs a clustering of size $\bigO(k_{\ell}(\Delta) \log n)$.
        \item In~\cref{sec:the_algorithm} we give the general algorithm.
        We choose some $\Delta' \in \Theta(\Delta)$. We iteratively construct an $(\ell, \Delta')$-clustering of size $\bigO(k_{\ell}(\Delta) \log n)$.
        Our greedy iterative algorithm maintains a set $C$ of pathlets and adds an $(\ell, \Delta')$-pathlet $(P, \I)$ to $C$ at every iteration.
        
        Consider having a set of pathlets $C = \{ (P_i, \I_i) \}$.
        We greedily select a pathlet $(P, \I)$ that covers as much of $\U \setminus \Cov(C)$ as possible, and add it to $C$. 
        Formally, we select a \emph{$(\Delta, \frac{1}{17})$-maximal} $(\ell, \Delta')$-pathlet: an $(\ell, \Delta')$-pathlet $(P, \I)$ such that
        \[
        | \Cov(P, \I) \setminus \Cov(C) | \geq \frac{1}{17} | \Cov(P', \I') \setminus \Cov(C) |
        \]
        for all $(\ell, \Delta)$-pathlets $(P', \I')$.        
        \item The subsequent goal is to compute $(\Delta, \frac{1}{17})$-maximal pathlets.
        We restrict pathlets to two types: those where the reference curve is 1) a vertex subcurve of $S$, or 2) a subsegment of an edge of $S$.
        Then we give algorithms for constructing pathlets of these types with a certain quality guarantee, i.e., pathlets that cover at least a constant fraction of what the optimal pathlet of that type covers.
        These algorithms are given in~\cref{sec:vertex-to-vertex_pathlets,sec:subedge_pathlets}.
    \end{enumerate}

            \subparagraph*{Reachability graph.} We introduce the \emph{reachability graph} in~\cref{sec:reachability_graph}.
            This graph is defined on a subcurve $W$ of $S$ and a set $Z$ of points in $\FSD(W, T)$.
            The reachability graph $G(W, T, Z)$ is a directed acyclic graph whose vertices are the set of points $Z$, together with certain boundary points of the free space $\FSD(W, T)$ and a collection of \emph{steiner points}.
            Given two points $(x, y)$ and $(x', y')$ in $Z$, the graph contains a directed path from $(x, y)$ to $(x', y')$ if and only if $d_F(W[x, x'], T[y, y']) \leq \Delta'$.

            We treat the free space diagram as a rectilinear polygon $\mathcal{R}$ with rectilinear holes, obtained by reducing all obstacles of $\FSD(W, T)$ to their intersections with the parameter space grid.
            We show that a bimonotone path between two points $p$ and $q$ exists in $\FSD(W, T)$ if and only if a rectilinear shortest path between $p$ and $q$ in $\mathcal{R}$ has length $\lVert p-q \rVert_1$, the $L_1$-distance between $p$ and $q$.
            The reachability graph $G(W, T, Z)$ is defined as the \emph{shortest paths preserving graph}~\cite{widmayer91rectilinear_graphs} for the set $Z$ with respect to $\mathcal{R}$, made into a directed graph by directing edges, which are all horizontal or vertical, to the right or top.
            This graph has $\bigO((|W|n + |Z|) \log (n|Z|))$ complexity, and a shortest path in the graph between points in $Z$ is also a rectilinear shortest path between the corresponding points in~$\mathcal{R}$.

            \subparagraph*{Vertex-to-vertex pathlets.} In~\cref{sec:vertex-to-vertex_pathlets} we construct a pathlet where the reference curve is a vertex subcurve of $S$.
            For a given vertex $S(i)$ of $S$, we construct reference-optimal $(\ell, \Delta')$-pathlets $(S[i, i+j], \I_j)$ for all $j \in [\ell]$.
            We first identify a set $Z$ of $\bigO(n \ell)$ \emph{critical points} in $\FSD(S[i, i+\ell], T)$.
            We show that for every reference curve $S[i, i+j]$, there is a reference-optimal $(\ell, \Delta')$-pathlet $(S[i, i+j], \I_j)$ where for each interval $[y, y'] \in \I_j$, the points $(i, y)$ and $(i+j, y')$ are critical points.
            We construct the intervals $\I_j$ through a sweepline algorithm over the reachability graph $G(S[i, i+\ell], T, Z)$, which has $\bigO(n \ell \log n)$ complexity.
            Our sweepline computes, for all $j \in [\ell]$, a reference-optimal $(j, \Delta')$-pathlet $(S[i, i+j], \I_j)$ by iterating over all in-edges to critical points $(i+j, y)$ in $G(S[i, i+\ell], T, Z)$.
            Doing this for all $i$ (and remembering the optimum)
            thereby takes $\bigO(n^2 \ell \log^2 n)$ time and $\bigO(n \ell \log n)$ space.

            \subparagraph*{Subedge pathlets.} In~\cref{sec:subedge_pathlets} we construct a pathlet where the reference curve is a subsegment of an edge of $S$.
            For a given edge $e$ of $S$, we again first identify a set $Z$ of $\bigO(n^2)$ \emph{critical points} in $\FSD(e, T)$.
            However, rather than restricting the intervals in pathlets based on these critical points, we restrict the reference curves based on these critical points.
            Specifically, there are $m = \bigO(n)$ unique $x$-coordinates of points in $Z$, which we order as $x_1, \dots, x_m$.
            We show that by allowing for pathlets to use subsegments of the reversal $\rev{e}$ of $e$ as reference curves, we may restrict reference curves to be of the form $e[x_i, x_{i'}]$ or $\rev{e}[x_i, x_{i'}]$ to not lose much coverage.
            That is, the optimal $(2, \Delta')$-pathlet with such a reference curve covers at least one-fourth of what any other $(2, \Delta')$-pathlet using a subsegment of $e$ as a reference curve covers.

            The remainder of our subedge pathlet construction algorithm follows the same procedure as for vertex-to-vertex pathlets, though with the following optimization.
            We consider every $x_i$ separately, for $i \in [m]$.
            However, rather than considering all reference curves $e[x_i, x_{i'}]$, of which there are $m-i$, we consider only $\bigO(\log (m-i))$ reference curves.
            The main observation is that we may split a pathlet $(e[x_i, x_{i'}], \I)$ into two: $(e[x_i, x_{i+2^j}], \I_1)$ and $(e[x_{i'-2^j}, x_{i'}], \I_2)$, for some $j \leq \log (m-i)$.
            One of the two pathlets covers at least half of what $(e[x_i, x_{i'}], \I)$ covers, so an optimal $(2, \Delta')$-pathlet $(e[x_i, x_{i+2^j}], \I)$ that is defined by critical points covers at least one-eighth of any other subedge $(2, \Delta')$-pathlet $(e[x, x'], \I')$.

            For every $i \in [m]$, we let $Z_i \subseteq Z$ be the subset of critical points with $x$-coordinate equal to $x_i$ or $x_{i+2^j}$ for some $j \leq \log (m-i)$.
            We construct the reachability graph $G(e, T, Z_i)$, which has $\bigO(n \log^2 n)$ complexity.
            We then proceed as with the vertex-to-vertex pathlets, using a sweepline through the reachability graph.
            Doing this for all $i$ (and remembering the optimal pathlet) thereby takes $\bigO(n^2 \log^3 n)$ total time and $\bigO(n \log^2 n)$ space.
            Taken over all edges of $S$, we obtain a subedge pathlet in $\bigO(n^3 \log^3 n)$ time and $\bigO(n \log^2 n)$ space.

\section{Pathlet-preserving simplifications}
\label{sec:simplification}

    We first aim to limit our attention to $(\ell, 4\Delta)$-pathlets $(P, \I)$ whose reference curves $P$ are subcurves of some universal curve $S$.
    This way, we may design an algorithm that considers all subcurves of $S$, rather than all curves in $\R^d$.
    This has the additional benefit of allowing the use of the free space diagram $\FSD[4\Delta](S, T)$ to construct pathlets, as seen in~\cref{fig:pathlets_in_free_space}.

    \begin{figure}
        \centering
        \includegraphics[page=2]{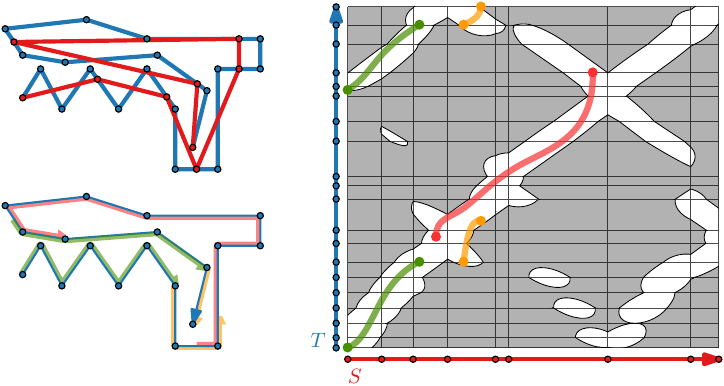}
        \caption{ \textbf{Top left:} A simplification $S$ (red) of the trajectory $T$ (blue).
        \textbf{Right:} The diagram $\FSD(S, T)$ in white.
        The obstacles of the diagram are colored in gray.
        The clustering (bottom left)  corresponds to a set of colored bimonotone paths, where paths of a given color are horizontally aligned, and the paths together span the entire vertical axis.}
        \label{fig:pathlets_in_free_space}
    \end{figure}
    
    For any $(\ell, \Delta)$-pathlet $(P, \I)$ there exists an $(n, 2 \Delta)$-pathlet $(P', \I)$ where $P'$ is a subcurve of $T$. Indeed, consider any interval $[a, b] \in \I$ and choose $P' = T[a, b]$. However, restricting the subcurves of $T$ to have complexity at most $\ell$ may significantly reduce the maximum coverage, see for example \cref{fig:optimum_arbitrary_complexity}.
    \begin{figure}[b]
        \centering
        \includegraphics{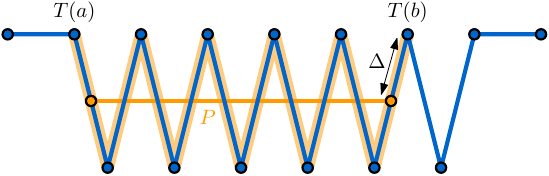}
        \caption{There exists a segment $P$ where $\dF(P, T[a, b]) \leq \Delta$. 
        In contrast, for any vertex-restricted $S$ with $\dF(T[a, b], S) \leq \Delta$, the complexity of $S$ is $\Theta(|T[a, b]|)$.}
        \label{fig:optimum_arbitrary_complexity}
    \end{figure}
    Instead of restricting pathlets to be subcurves of $T$, we restrict them to be subcurves of a different curve $S$. We enforce the following property: 
        
    \begin{definition}
    \label{def:simplification}
        For a trajectory $T$ and value $\Delta \geq 0$, a pathlet-preserving simplification is a curve $S$ together with a $2\Delta$-matching $(f, g)$, where for any subtrajectory $T[a, b]$ of $T$ and all curves $P$ with $\dF(P, T[a, b]) \leq \Delta$, the subcurve $S[s, t]$ matched to $T[a, b]$ by $(f, g)$ has complexity $|S[s, t]| \leq |P| + 2 - |\N \cap \{s, t\}|$.
    \end{definition}

    \begin{theorem}
        \label{thm:property_S}
        Let $(S, f, g)$ be a pathlet-preserving simplification of $T$. 
        For any $(\ell, \Delta)$-pathlet $(P, \I)$, there exists a subcurve $S[s, t]$ such that $(S[s, t], \I)$ is an $(\ell +2 - |\N \cap \{s, t\}|, 4\Delta)$-pathlet.
    \end{theorem}
    \begin{proof}
            Consider any $(\ell, \Delta)$-pathlet $(P, \I)$ and choose some interval $[a, b] \in \I$. 
            For all $[c, d] \in \I$, via the triangle inequality, $\dF(T[a, b], T[c, d]) \leq 2 \Delta$.
            Let $S[s, t]$ be the subcurve of $S$ matched to $T[a, b]$ by $(f, g)$.
            Naturally, $\dF(S[s, t], T[a, b]) \leq 2\Delta$, and so by the triangle inequality $\dF(S[s, t], T[c, d]) \leq 4\Delta$.
            By the definition of a pathlet-preserving simplification, we obtain that for every curve $P'$ with $\dF(P', T[a, b]) \leq \Delta$, we have $|P'| \geq |S[s, t]| - 2 + |\N \cap \{s, t\}|$.
            In particular, setting $P' \gets P$ implies that $|S[s, t]| \leq \ell + 2 - |\N \cap \{s, t\}|$.
            Thus $(S[s, t], \I)$ is an $(\ell + 2 - |\N \cap \{s, t\}|, 4\Delta)$-pathlet.
    \end{proof}
        
\subparagraph*{Prior simplifications.}
    The curve $S$ that we construct is a \emph{curve-restricted $\alpha \Delta$-simplification} of $T$; a curve whose vertices lie on $T$, where for every edge $s = \overline{T(a) T(b)}$ of $S$ we have $\dF(s, T[a, b]) \leq \alpha \Delta$. 
    Various $\alpha \Delta$-simplification algorithms have been proposed~\cite{guibas93minimum_link, agarwal_near-linear_2005, de_berg_fast_2013, van2018global}.
    
    If $T$ is a curve in $\mathbb{R}^2$, Guibas~\etal~\cite{guibas93minimum_link} provide an $\bigO(n \log n)$ time algorithm that constructs a $2 \Delta$-simplification $S$ for which there is no $\Delta$-simplification $S'$ with $|S'| < |S|$.
    Their algorithm is not efficient in higher dimensions however.
    
    Agarwal~\etal~\cite{agarwal_near-linear_2005} also construct a $2\Delta$-simplification $S$ of $T$ in $\bigO(n \log n)$ time. This was applied by Akitaya~\etal~\cite{bruning_subtrajectory_2023} for their subtrajectory clustering algorithm under the discrete \f distance. 
    The simplification $S$ has a similar guarantee as the simplification of~\cite{guibas93minimum_link}: there exists no \emph{vertex-restricted} $\Delta$-simplification $S'$ with $|S'| < |S|$.
    This guarantee is weaker than that of~\cite{guibas93minimum_link}, as vertex-restricted simplifications are simplifications formed by taking a subsequence of vertices of $T$ as the vertices of the simplification.
    It can, however, be constructed efficiently in higher dimensions.
    
    As we show in~\cref{fig:optimum_arbitrary_complexity}, the complexity of a vertex-restricted $\Delta$-simplification can be arbitrarily bad compared to the (unrestricted) $\Delta$-simplification with minimum complexity. 
    Br\"{u}ning~\etal~\cite{bruning_faster_2022} note that for the subtrajectory problem under the continuous \f distance, one requires an $\alpha \Delta$-simplification whose complexity has guarantees with respect to the optimal (unrestricted) simplification.
    They present a $3 \Delta$-simplification $S$ (whose definition was inspired by de Berg, Gudmundsson and Cook~\cite{de_berg_fast_2013}) with the following property:
    for any subcurve $T[a, b]$ of $T$ within \f distance $\Delta$ of some line segment, there exists a subcurve $S[s, t]$ of $S$ with complexity at most $4$ that has \f distance at most $3\Delta$ to $T[a, b]$.
    Thus, there exists no $\Delta$-simplification $S'$ with $|S'| < |S|/2$. 
        
\subparagraph*{Our new curve simplification.}
    In Definition~\ref{def:simplification} we presented yet another curve simplification under the \f distance for curves in $\R^d$.  
    Our simplification has a stronger property than the one that is realized by Br\"{u}ning~\etal~\cite{bruning_faster_2022}: 
    for any subcurve $T[a,b]$ and \emph{any} curve $P$ with $\dF(P, T[a, b]) \leq \Delta$, we require that there exists a subcurve $S[s, t]$ with $\dF(S[s, t], T[a, b]) \leq 2\Delta$ that has at most two more vertices than $P$.
    This implies both the property of Br\"{u}ning~\etal~\cite{bruning_faster_2022} and ensures that no $\Delta$-simplification $S'$ exists with $|S'| < |S| - 2$.

    In Appendix~\ref{app:constructing_pathlet_preserving_simplification}, we provide an efficient algorithm for constructing pathlet-preserving simplifications.
    We relegate this section to the appendix to not distract from the main storyline.
    The algorithm is an extension of the vertex-restricted simplification of Agarwal~\etal~\cite{agarwal_near-linear_2005} to construct a curve-restricted simplification instead.
    For this, we use the techniques of Guibas~\etal~\cite{guibas93minimum_link} to quickly identify if an edge of $T$ is suitable to place a simplification vertex on.
    We combine this check with the algorithm of~\cite{agarwal_near-linear_2005} and obtain:

    \begin{restatable}{theorem}{thmSimplification}
    \label{thm:simplification}
        For any trajectory $T$ with $n$ vertices and any $\Delta \geq 0$, we can construct a pathlet-preserving simplification $S$ in $\bigO(n \log n)$ time. 
    \end{restatable}

\section{The universe $\U$ and greedy set cover}
\label{sec:set_cover}

 Subtrajectory clustering is closely related to the \emph{set cover} problem.
    In this problem, we have a discrete universe $\U$ and a family of sets $\mathcal{S}$ in this universe, and the goal is to pick a minimum number of sets in $\mathcal{S}$ such that their union is the whole universe.
    The decision variant of set cover is NP-complete~\cite{karp72reducibility}.
    However, the following greedy strategy gives an $\bigO(\log |\U|)$ approximation of the minimal set cover size~\cite{chvatal79greedy}.
    Suppose we have picked a set $\hat{\mathcal{S}} \subseteq \mathcal{S}$ that does not yet cover all of $\U$.
    The idea is then to add a set $S \in \mathcal{S}$ that maximizes $|S \cap (\U \setminus \bigcup \hat{\mathcal{S}})|$, and to repeat the procedure until $\U$ is fully covered.

    \subparagraph{Defining the universe $\U$.}
    We apply this greedy strategy to subtrajectory clustering, putting the focus on constructing a pathlet that covers the most of some universe $\U$.
    For subtrajectory clustering, the universe is, in principle, infinite.
    We therefore first define a discrete universe $\U$ consisting of $\bigO(n^3)$ intervals that together cover $[1, n]$.
    We choose this universe carefully, as an optimal covering of $\U$ with pathlets must have roughly the same size as an optimal covering of $[1, n]$.
    We define $\U$ using the following set of \emph{critical points} in $\FSD(S, T)$:

    \begin{definition}
    \label{def:critical_points}
        For $i \in [|S|-1]$ and $j \in [n-1]$, consider their corresponding cell (the area $[i, i+1] \times [j, j+1]$) and the following six extreme points:       
        \begin{itemize}
            \item A leftmost point of $\FSD(S, T) \cap ([i, i+1] \times [j, j+1])$,
            \item A rightmost point of $\FSD(S, T) \cap ([i, i+1] \times [j, j+1])$,
            \item The leftmost and rightmost points of $\FSD(S, T) \cap ([i, i+1] \times \{j\})$, and
            \item The leftmost and rightmost points of $\FSD(S, T) \cap ([i, i+1] \times \{j+1\})$.
        \end{itemize}
        Let $X_{i, j}$ be the set of corresponding $x$-coordinates and $X := \bigcup\limits_{i, j} X_{i, j}$.
        For each $x \in X$, we call  every point $(x, y)$ that is an endpoint of a connected component (vertical segment) of $\FSD(S, T) \cap (\{x\} \times [1, n])$ a \emph{critical point}.
    \end{definition}

    \begin{definition}
        Let $Y^*$ be the set of critical points, sorted by their $y$-coordinate. 
        We define the set $\U$ as the set of intervals in $[1, n]$ between two consecutive $y$-coordinates in $X$.
        Since there are at most $6n$ critical points in $[i, i+1] \times [j, j+1]$ for each $i \in [|S|-1]$ and $j \in [n-1]$, it follows that $|\U| \leq 6n^3-1 = \bigO(n^3)$. 
    \end{definition}

    \begin{lemma}
    \label{lem:constructing_universe}
        We can construct $\U$ in $\bigO(n^3)$ time.
    \end{lemma}
    \begin{proof}
        Fix an integer $i \in [|S|-1]$.
        We compute the critical points inside the cells $[i, i+1] \times [j, j+1]$, for all $j \in [n-1]$, in $\bigO(n^2)$ time altogether.
        For this, we compute the sets $X_{i, j}$ of~\cref{def:critical_points} in $\bigO(1)$ time each.
        Let $X_i = \bigcup_j X_{i, j}$.
        Then, we compute the intersections of each vertical line $\{x\} \times [1, n]$, for $x \in X_i$, in $\bigO(n)$ time each.
        The critical points inside the cells $[i, i+1] \times [j, j+1]$, for $j \in [n-1]$, are the endpoints of connected components of these intersections, and can be computed in $\bigO(n)$ time per line, totalling $\bigO(n^2)$ time.
        Summing over all integers $i$ completes the proof.
    \end{proof}

    \subparagraph{Applying greedy set cover.}
    In the remainder of this paper, we let $\U$ denote this discrete universe.
    We generalize the analysis of the greedy set cover argument to pathlets that cover a (constant) fraction of what the optimal pathlet covers.
    This relaxes the requirements on the pathlets and helps reduce complexity of the problem.
    For this, we introduce the following:

    \begin{definition}[Maximal pathlets]
        Given a set $C$ of pathlets, a \emph{$(\Delta, \frac{1}{c})$-maximal} $(\ell, \Delta')$-pathlet $(P', \I')$ is a pathlet such that there exists no $(\ell, \Delta)$-pathlet $(P, \I)$ with 
        \[
            \frac{1}{c} | \Cov(P, \I) \setminus \Cov(C) |  \geq | \Cov(P', \I') \setminus \Cov(C) |.
        \]
    \end{definition}

    \noindent
    In~\cref{lem:optimalanswer}, we show that if we keep greedily selecting $(\Delta, \frac{1}{c})$-maximal pathlets for our clustering, the size of the clustering stays relatively small compared to the optimum size.
    The bound closely resembles the bound obtained by the argument for greedy set cover.

    \begin{lemma}
    \label{lem:optimalanswer}
        Iteratively adding $(\Delta, \frac{1}{c})$-maximal pathlets yields a clustering of size at most $3c \cdot k_\ell(\Delta) \ln (6n) + 1$.
    \end{lemma}
    \begin{proof}
        Let $C^* = \{ (P_i, \I_i) \}_{i=1}^k$ be an $(\ell, \Delta)$-clustering of $T$ of minimal size. 
        Then $k := |C^*| = k_{\ell}(\Delta)$. 
        Consider iteration $j$ of the algorithm, where we have some set of $(\ell, \Delta')$-pathlets $C_j$.
        Denote by $W_j = |\U| \setminus \Cov(C_j)$ the ``size'' of the part of the universe that still needs to be covered.
        Since $C^*$ covers $\U$, it must cover $\U \setminus \Cov(C_j)$. 
        It follows via the pigeonhole principle that there is at least one $(\ell,\Delta)$-pathlet $(P_i, \I_i) \in C^*$ that covers at least $W_j / k$ intervals in $\Cov(P_i, \I_i) \setminus \Cov(C_j)$.
        Per definition of being $(\Delta, \frac{1}{c})$-maximal, our greedy algorithm finds a pathlet $(P_j, \I_j)$ that covers at least $\frac{W_j}{ck}$ uncovered intervals.
        Thus:
        \[
            W_{j+1} = |\U| - | \Cov(C_j) \cup \Cov(P_j, \I_j) |
            \leq W_j - \frac{W_j}{c \cdot k} = W_j \cdot ( 1 - \frac{1}{c \cdot k} ).
        \]
        
        We have that $W_0 = |\U|$.
        Suppose it takes $k'+1$ iterations to cover all of $T'$ with the greedy algorithm.
        Then before the last iteration, at least one edge of $T'$ remained uncovered.
        That is, $|\U| \cdot \left( 1-\frac{1}{c \cdot k} \right)^{k'} \geq 1$.
        We apply that $e^x \geq 1+x$ for all real $x$ to obtain: 
        \[
            \frac{1}{e} \geq \left( 1-\frac{1}{x} \right)^x
        \]
        for all $x \geq 1$.
        Plugging in $x \gets c \cdot k$, it follows that
        \[
            1 \leq |\U| \cdot \left( 1 - \frac{1}{c \cdot k} \right)^{k'}
            = |\U| \cdot \left( 1 - \frac{1}{c \cdot k} \right)^{c \cdot k \cdot \frac{k'}{c \cdot k}}
            \leq |\U| \cdot e^{-\frac{k'}{c \cdot k}}.
        \]
        Hence $e^{\frac{k'}{c \cdot k}} \leq |\U|-1$, showing that $k' \leq c \cdot k \ln (|\U|-1)$.
        Thus after $k'+1 \leq c \cdot k_\ell(\Delta) \ln (|\U|-1) + 1$ iterations, all of $T'$, and therefore $T$, is covered.
        Using that $|\U| \leq 6n^3-1$ completes the proof.
    \end{proof}

\section{Subtrajectory clustering}
\label{sec:the_algorithm}

    In this section we present our algorithm for subtrajectory clustering.
    We first restrict our attention to reference curves of two types.
    
    Recall that using the pathlet-preserving simplification $S$ of $T$, we may already restrict our attention to reference curves that are subcurves of $S$.
    Still, the space of possible reference curves remains infinite.
    We wish to discretize this space by identifying certain finite classes of reference curves that contain a ``good enough'' reference curve, i.e., one with which we can construct a pathlet that is $(\Delta, \frac{1}{c})$-maximal for some small constant $c$.

    We distinguish between two types of pathlets, based on their reference curves (note that not all pathlets fit into a class, and that some may fit into both classes):
    \begin{enumerate}
        \item Vertex-to-vertex pathlets: pathlets $(P, \I)$ where $P$ is a vertex subcurve of $S$.
        \item Subedge pathlets: pathlets $(P, \I)$ where $P$ is a subsegment of an edge of~$S$.
    \end{enumerate}

    \begin{figure}
        \centering
        \includegraphics[page=3]{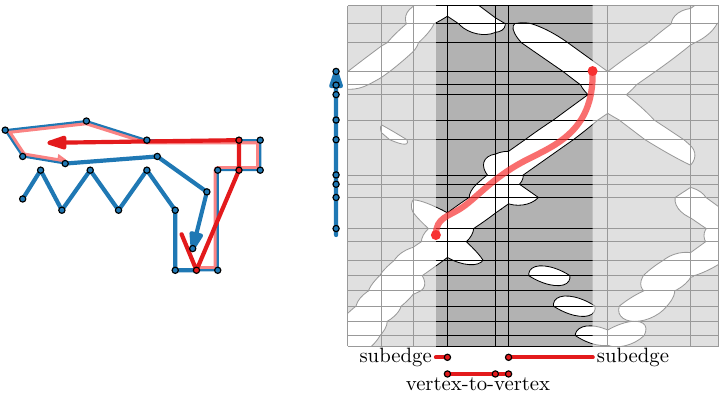}
        \caption{A pathlet (left), corresponding to the red $\Delta'$-matching (right), gets split into a vertex-to-vertex and two subedge pathlets.
        The new pathlets correspond to the parts of the red matching that are vertically above the part of the $x$-axis corresponding to the new reference curve.
        }
        \label{fig:splitting_pathlets}
    \end{figure}

\noindent
    We construct pathlets of the above types, ensuring that they all cover at least some constant fraction of the optimal coverage for pathlets of the same type.
    Let $(P_\ver, \I_\ver)$ and $(P_\sub, \I_\sub)$ respectively be a vertex-to-vertex and subedge $(\ell, \Delta')$-pathlet, that respectively cover at least a factor $\frac{1}{c_\ver}$ and $\frac{1}{c_\sub}$ of an optimal pathlet of the same type.
    We show that one of these two pathlets is a $(\Delta, \frac{1}{c})$-maximal pathlet, for $c = c_\ver + 2c_\sub$.
    For intuition, refer to~\cref{fig:splitting_pathlets}.

    \begin{restatable}{lemma}{lemSplittingPathlets}
    \label{lem:quality_of_pathlet}
        Given a collection $C$ of pathlets, let
        \[
            (P, \I) \in \{(P_\ver, \I_\ver), (P_\sub, \I_\sub)\}
        \]
        be a pathlet with maximal coverage among the uncovered points.
        Then $(P, \I)$ is $(\Delta, \frac{1}{c})$-maximal with respect to $C$, for $c = c_\ver + 2c_\sub$.
    \end{restatable}

    \begin{proof}
        Let $(P^*, \I^*)$ be an optimal $(\ell, \Delta)$-pathlet.
        By~\cref{thm:property_S}, there exists a subcurve $S[x, x']$ of $S$ such that $(S[x, x'], \I^*)$ is a $(\ell + 2 - |\N \cap \{x, x'\}|, \Delta')$-pathlet.
        Suppose first that $S[x, x']$ is a subsegment of an edge of $S$, making $(S[x, x'], \I^*)$ a subedge pathlet with $|S[x, x']| \leq 2$.
        In this case, the coverage of $(P_\sub, \I_\sub)$ is at least $\frac{1}{c_\sub}$ times the coverage of $(S[x, x'], \I^*)$ over the uncovered points.
        Hence $(P_\sub, \I_\sub)$ is $(\Delta, \frac{1}{c_\sub})$-maximal.
        Since the pathlet $(P, \I)$ has at least as much coverage as $(P_\sub, \I_\sub)$, it must also be $(\Delta, \frac{1}{c_\sub})$-maximal.

        Next suppose that $S[x, x']$ is not a subsegment of an edge of $S$, meaning $S[x, x']$ contains at least one vertex of $S$.
        In this case, we split $S[x, x']$ into three subcurves:
        \begin{itemize}
            \item A suffix $P_\suf = S[x, \ceil{x}]$ of an edge,
            \item A vertex subcurve $P_\ver = S[\ceil{x}, \floor{x'}]$, and
            \item A prefix $P_\pre = S[\floor{x'}, x']$ of an edge.
        \end{itemize}
        Observe that every subcurve has at most $\ell$ vertices.
        The suffix and prefix both trivially have at most $2 \leq \ell$ vertices.
        The vertex subcurve has at most the number of vertices of $S[x, x']$, but if $x$, respectively $x'$, is not an integer, then the vertex subcurve loses a vertex compared to $S[x, x']$.
        That is, the vertex subcurve has at most
        \[
            \ell + 2 - |\N \cap \{x, x'\}| - |\{x, x'\} \setminus \N| = \ell + 2 - |\{x, x'\}| = \ell \quad \textnormal{ vertices. }
        \]
        
        Since every interval $[y, y'] \in \I^*$ corresponds to a $\Delta'$-matching $M$ between $S[x, x']$ and $T[y, y']$, we can decompose $[y, y']$ into three intervals $[y, y_1]$, $[y_1, y_2]$ and $[y_2, y']$, such that $M$ decomposes into three $\Delta'$-matchings, one between $P_\suf$ and $T[y, y_1]$, one between $P_\ver$ and $T[y_1, y_2]$, and one between $P_\pre$ and $T[y_2, y']$.
        By decomposing all intervals in $\I^*$ in this manner, we obtain that there are three $(\ell, \Delta)$-pathlets $(P_\suf, \I^*_\suf)$, $(P_\ver, \I^*_\ver)$ and $(P_\pre, \I^*_\pre)$ that together have the same coverage as $(P^*, \I^*)$.

        We have at least one of the following:
        \begin{itemize}
            \item $(P_\suf, \I^*_\suf)$ covers at least a factor $\frac{c_\sub}{c_\ver + 2c_\sub}$ of what $(P^*, \I^*)$ covers, or
            \item $(P_\ver, \I^*_\ver)$ covers at least a factor $\frac{c_\ver}{c_\ver + 2c_\sub}$ of what $(P^*, \I^*)$ covers, or
            \item $(P_\pre, \I^*_\pre)$ covers at least a factor $\frac{c_\sub}{c_\ver + 2c_\sub}$ of what $(P^*, \I^*)$ covers.
        \end{itemize}
        Regardless of what statement holds, the pathlet $(P, \I)$ covers at least a factor $\frac{1}{c_\ver + 2c_\sub}$ of what $(P^*, \I^*)$ covers.
        Thus we have that $(P, \I)$ is $(\Delta, \frac{1}{c_\ver + 2c_\sub})$-maximal.
    \end{proof}

    Next we combine the previous ideas on simplification and greedy algorithms and present our algorithm for subtrajectory clustering.
    The algorithm uses subroutines for constructing the two types of pathlets described above, as well as a data structure for comparing their coverages to select the best pathlet for the clustering.
    
    Our pathlet construction algorithms guarantee that $c_\ver = 1$ and $c_\sub = 8$.
    By~\cref{lem:quality_of_pathlet}, the pathlet with the most coverage is therefore $(\Delta, \frac{1}{c})$-maximal with respect to the uncovered points, for $c = 1 + 2 \cdot 8 = 17$.
    By~\cref{lem:optimalanswer}, the resulting $(\ell, \Delta')$-clustering has a size of at most $17k_\ell(\Delta) \ln (|\U|-1) + 1$.
    Since our universe $\U$ has size at most $6n^3-1$, the clustering has a size of at most $51k_\ell(\Delta) \ln (6n) + 1$.

\subparagraph*{A data structure for comparing pathlets.}
    Recall that we fixed some discrete universe $\U$ of $\bigO(n^3)$ intervals, and that we denote $\Cov(P, \I) = \Cov_\U(P, \I)$.
    In each iteration of our greedy algorithm, we select one of two pathlets whose  coverage is the maximum over $\U \setminus \Cov(C)$, given the current set of picked pathlets $C$.
    To compare the coverages of pathlets, we make use of binary search trees built on $\U$ and $\Cov(C)$:

    \begin{restatable}{lemma}{coverageDS}
    \label{lem:coverage_ds}
        In $\bigO(n^3 \log n)$ time, we can preprocess $\U$ and $\Cov(C)$ into a data structure of $\bigO(n^3)$ size, such that given a pathlet $(P, \I)$, the value $| \Cov(P, \I) \setminus \Cov(C) |$ can be computed in $\bigO(|\I| \log n)$ time.
    \end{restatable}
    \begin{proof}
        We make use of a general data structure for storing a set $\I$ of $m$ interior-disjoint intervals, such that given a query interval $I$, the number of intervals in $\I$ that are fully contained in $I$ can be reported efficiently.
        For the data structure, we store the (multiset of) endpoints of intervals in $\I$ in a balanced binary search tree.
        The tree uses $\bigO(m)$ space and is constructed in $\bigO(m \log m)$ time.
        
        We report the number of intervals in $\I$ contained in a query interval $I$ as follows.
        An interval $[a, b] \in \I$ is contained in $I$ if and only if both $a$ and $b$ are.
        Furthermore, there are $k' \leq 2$ intervals in $\I$ that $I$ intersects but does not contain.
        Thus, if $I$ contains $k$ endpoints stored in the binary search tree, then it contains $(k-k') / 2$ intervals of $\I$.
        We compute $k'$ by reporting the intervals of $\I$ containing the endpoints of $I$ in $\bigO(\log m)$ time.
        Computing $k$ and then reporting $(k-k') / 2$ takes an additional $\bigO(\log m)$ time.
        Thus we answer a query in $\bigO(\log m)$ time.

        We use the above data structure to efficiently compute $| \Cov(P, \I) \setminus \Cov(C) |$ for a query pathlet $(P, \I)$.
        For this, we preprocess both $\U$ and $\Cov(C)$ into the above data structure.
        Since $\Cov(C) \subseteq \U$ and $|\U| = \bigO(n^3)$, this takes $\bigO(n^3 \log n)$ time, and the data structures use $\bigO(n^3)$ space.
        With the two data structures, we report the values $| \Cov_\U(P, \I) \cap \U |$ and $| \Cov(P, \I) \cap \Cov(C) |$ in $\bigO(\log n)$ time.
        We then report
        \[
            | \Cov(P, \I) \setminus \Cov(C) | = | \Cov(P, \I) \cap \U | - | \Cov(P, \I) \cap \Cov(C) |.\qedhere
        \]
    \end{proof}

\subparagraph*{Asymptotic complexities.}

    Our algorithm iteratively constructs a set $C$ of $\bigO(k_\ell(\Delta) \log n)$ pathlets.
    Before we start constructing pathlets, we compute the universe $\U$ of $\bigO(n^3)$ intervals.
    This takes $\bigO(n^3)$ time~(\cref{lem:constructing_universe}).
    
    In each iteration, we construct the data structure of~\cref{lem:coverage_ds} on the universe $\U$ and current set of pathlets $C$.
    This takes $\bigO(n^3 \log n)$ time and uses $\bigO(n^3)$ space.
    Constructing the vertex-to-vertex pathlet then takes $\bigO(n^2 \ell \log^2 n)$ time and uses $\bigO(n \ell \log n)$ space~(\cref{thm:constructing_vertex-to-vertex}).
    The subedge pathlet takes $\bigO(n^3 \log^3 n)$ time and $\bigO(n \log^2 n)$ space to construct~(\cref{thm:constructing_subedge}).
    
    To decide which pathlet to use in the clustering, we make further use of the data structure of~\cref{lem:coverage_ds}.
    All constructed pathlets $(P, \I)$ have $|\I| \leq n$, and so we compute the coverages of the two pathlets in $\bigO(n \log n)$ time.
    By summing up all complexities, we derive our main theorem:
    
    \begin{theorem}
    \label{thm:yes-no-no}
        Given a trajectory $T$ with $n$ vertices, an integer $\ell \geq 2$, and  a value $\Delta \geq 0$, we can construct an $(\ell, 4\Delta)$-clustering of size at most $51 k_\ell(\Delta) \ln (6n) + 1$ in $\bigO(k_\ell(\Delta) n^3 \log^4 n)$ time and using $\bigO(n^3)$ space.
    \end{theorem}

\section{The reachability graph}
\label{sec:reachability_graph}

    Let $\Delta' = 4\Delta$.
    For any subcurve $W$ of $S$ and a set of points $Z$ in $\FSD(W, T)$ we define the \emph{reachability graph} $G(W, T, Z)$.
    The vertices of this graph are the set of points $Z$, together with some Steiner points in $[1, |W|] \times [1, |T|]$.
    The reachability graph $G(W, T, Z)$ is a directed graph where, for any two $\mu_1, \mu_2 \in Z$, there exists a directed path from $\mu_1$ to $\mu_2$ if and only if $\mu_2$ is reachable from $\mu_1$ in the free space $\FSD(W, T)$.
    
We define a reachability graph with $\bigO((n |W| + |Z|) \log (n|Z|))$ vertices and edges, and can be constructed in $\bigO((n |W| + |Z|) \log (n|Z|))$ time.

    \begin{restatable}{theorem}{thmReachabilityGraph}
    \label{thm:reachability_graph}
        Let $W$ be a subcurve of $S$ and $Z$ a set of points in $\FSD(W, T)$.
        The reachability graph $G(W, T, Z)$ has $\bigO((n |W| + |Z|) \log (n|Z|))$ vertices and edges, and can be constructed in $\bigO((n |W| + |Z|) \log (n|Z|))$ time.
    \end{restatable}

    \subparagraph{Constructing the graph.}
    \Cref{lem:simplified_FSD} shows that when focusing on reachability between points in $\FSD(W, T)$, we can simplify obstacles of the free space diagram to the parameter space grid, minus the free space.
    See~\cref{fig:polygon}.
   
    These simplified obstacles can be represented in $\bigO(|W| n)$ time as a set of horizontal and vertical line segments (whose endpoints are not included, except possibly some that meet the boundary of $[1, |W|] \times [1, |T|]$).
    The complement of these segments in the parameter space $[1, |W|] \times [1, |T|]$ gives a rectilinear polygon with rectilinear holes $\mathcal{R}$.

    \begin{figure}
        \centering
        \includegraphics{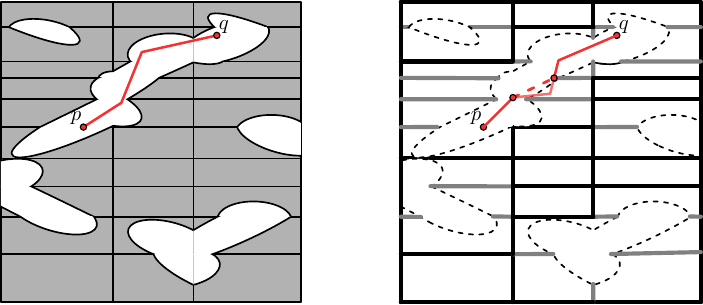}
        \caption{(left) The $\Delta'$-free space diagram of $W$ and $T$ with points $p$ and $q$ connected by a bimonotone path. 
        (right) The obstacles of $\mathcal{R}$ are made up of all grid edges that are entirely contained in the obstacles of $\FSD(W, T)$ (shown in black) plus the gray segments.
        We may transform any bimonotone path between $p$ and $q$ into one that lies in $\FSD(W, T)$.
        }
        \label{fig:polygon}
    \end{figure}

    \begin{lemma}
    \label{lem:simplified_FSD}
        Let $p$ and $q$ be two points in $\FSD(W, T)$.
        There is a bimonotone path from $p$ to $q$ in $\FSD(W, T)$ if and only if there is a bimonotone path from $p$ to $q$ in $\mathcal{R}$.
    \end{lemma}
    \begin{proof}
        Since $\FSD(W, T)$ is completely contained in $\mathcal{R}$, any path in $\FSD(W, T)$ is also a path in $\mathcal{R}$.
        To transform a path from $p$ to $q$ in $\mathcal{R}$ to a path in $\FSD(W, T)$, replace each maximal subpath $\pi$ that lies inside a cell of $\FSD(W, T)$, with the segment connecting its endpoints.
        The obstacles of $\mathcal{R}$ agree with the obstacles of $\FSD(W, T)$ on the boundary of cells, and thus if $\pi$ starts or ends on the boundary of a cell, the respective endpoint lies in $\FSD(W, T)$.
        Additionally, because $p$ and $q$ lie in $\FSD(W, T)$, we have that $\pi$ must always start and end at points in $\FSD(W, T)$.
        By convexity of the free space inside a cell, the line segment connecting the endpoints of $\pi$ lies in $\FSD(W, T)$, and so does the resulting path.
        This replacement preserves bimonotonicity, completing the proof.
    \end{proof}

    \noindent
    To obtain $G(W, T, Z)$ we first construct an undirected graph $G(Z)$.
    This graph is the \emph{shortest paths preserving graph} by Widmayer~\cite{widmayer91rectilinear_graphs}.
    The vertices of $G(Z)$ are the points in $Z$, together with the vertices of $\mathcal{R}$ and some Steiner points.
    By weighting each edge by its length, the graph perfectly encodes rectilinear distances between points in $Z$.
    That is, the rectilinear distance in $\mathcal{R}$ between two points in $Z$ is equal to their distance in~$G(Z)$.
    
    The number of vertices of $\mathcal{R}$ is $\bigO(n |W|)$, giving the graph a complexity of $\bigO((n |W| + |Z|) \log (n|Z|))$.
    The graph can be constructed in $\bigO((n |W| + |Z|) \log (n|Z|))$ time~\cite{widmayer91rectilinear_graphs}.
    
    The edges of $G(Z)$ are all horizontal or vertical line segments.
    We set $G(W, T, Z)$ to be the graph $G(Z)$, but with each edge directed towards the right (if horizontal) or top (if vertical).
    Observe that $G(W, T, Z)$ perfectly encodes reachability: for two points $p = (x, y)$ and $q = (x', y')$ in $Z$, if there is a bimonotone rectilinear path from $p$ to $q$ in $\mathcal{R}$, then any rectilinear shortest path from $p$ to $q$ must be bimonotone, and hence there must be a (bimonotone) path between them in $G(W, T, Z)$.
    Conversely, any path in $G(W, T, Z)$ is also a path in $\mathcal{R}$.
    Thus $\dF(P[x, x'], T[y, y']) \leq \Delta'$ if and only if there is a (bimonotone) path from $(x, y)$ to $(x', y')$ in~$G(W, T, Z)$.

    \thmReachabilityGraph*

\section{Vertex-to-vertex pathlets}
\label{sec:vertex-to-vertex_pathlets}

    Let $\Delta' = 4\Delta$, and let $C$ be a set of pathlets.
    Recall that we can compute $|\Cov(P, \I) \setminus \Cov(C)|$, for a given pathlet $(P, \I)$, in $\bigO(|\I| \log n)$ time (Lemma~\ref{lem:coverage_ds}).
    We fix some integer $i$. We then give an algorithm for constructing a vertex-to-vertex $(\ell, \Delta')$-pathlet $(P, \I)$ where $P$ starts at the $i$'th vertex of $S$, and its coverage over $\U \setminus \Cov(C)$ is maximum.
    
    We find for each subcurve $S'$ of $S$ of length at most $\ell$ a reference-optimal $(\ell, \Delta')$-pathlet.
    To this end, we consider each vertex $S(i)$ of $S$ separately. 
    We construct a set of reference-optimal pathlets $(S[i, i+1], \I_1), \dots, (S[i, i+j], \I_j), \dots, (S[i, i+\ell], \I_\ell)$.
    We let each interval $\I_j$ contain all maximal intervals $[y, y']$ for which $\dF(S[i, i+j], T[y, y']) \leq \Delta'$, and thus all maximal intervals for which $(i, y)$ can reach $(i+j, y')$ by a bimonotone path in $\FSD(S, T)$.

    Recall that in Definition~\ref{def:critical_points} we defined a set of  \emph{critical points}. Let $Z$ denote all critical points in $\FSD(S[i, i+\ell], T)$ of the form $(i+j, y)$, for integers $j \in [\ell]$.
    That is, $Z$ contains for all  $j \in [\ell]$ the endpoints of all connected components (vertical line segments) of $\FSD(S, T) \cap (\{i+j\} \times [1, n])$.    
    Since each cell has at most $\bigO(1)$ such critical points, it follows that $|Z| \in \bigO(n \ell)$.
    Observe that for any $\Delta'$-matching between $T$ and a vertex-to-vertex subcurve, we can always extend each curve in the matching to start and end at a point in $Z$: 

    \begin{observation}
   Let $(P, \I)$ be a vertex-to-vertex  $(\ell, \Delta')$-pathlet where $P$ starts at the $i$'th vertex. 
        Then there exists an $(\ell, \Delta')$-pathlet $(P, \I')$ with $\Cov(P, \I) \subseteq \Cov(P, \I')$ such that for each interval in $\I'$, the corresponding bimonotone path in $\FSD(S, T)$ starts and ends at a point in $Z$. 
    \end{observation}

    We create a sweepline algorithm that, for each $j \in [\ell]$, constructs a reference-optimal $(\ell, \Delta')$-pathlet $(S[i, j], \I_j)$. 
    We let each interval $\I_j$ contain all maximal intervals $[y, y']$ for which $\dF(S[i, i+j], T[y, y']) \leq \Delta'$, and thus all maximal intervals for which $(i, y)$ can reach $(i+j, y')$ by a bimonotone path in $\FSD(S, T)$.
    Note that both $(i, y)$ and $(i+j, y')$ are critical points.
    Thus we aim to find all maximal intervals $[y, y']$ for which $\FSD(S, T)$ contains a bimonotone path between critical points $(i, y)$ and $(i+j, y')$.
    
    To this end, we construct, for each $i \in [n]$, the reachability graph $G(S[i, i+\ell], T, Z)$ from~\cref{sec:reachability_graph}, which encodes reachability between all critical points.
    This graph takes $\bigO((n \ell + |Z|) \log (n|Z|)) = \bigO(n \ell \log n)$ time to construct and has complexity $\bigO(n \ell \log n)$ (see~\cref{thm:reachability_graph}).
    We aim to annotate each vertex $\mu$ (which does not necessarily have to be a critical point) in $G(S[i, i+\ell], T, Z)$ with the minimum $y$, such that there exists a critical point $(i, y)$ that can reach $\mu$.
    We annotate $\mu$ with $\infty$ if no such value $y$ exists.
    
\subparagraph*{Annotating vertices.}
    We begin by annotating the vertices $(i, y)$ in $\bigO(n)$ time, by scanning over them in order of increasing $y$-coordinate.
    We go over the remaining vertices in $yx$-lexicographical order, where we go over the vertices based on increasing $y$-coordinate, and increasing $x$-coordinate when ties arise.
    Each vertex $\mu$ that we examine has only incoming arcs originating from vertices below and left of $\mu$.
    By our lexicographical ordering, each of these vertices are already annotated. 
    The minimal $y$ for which there exists a path from $(i, y)$ to $\mu$, must be the minimum over all its incoming arcs which we compute in time proportional to the in-degree of $\mu$. 
    If $\mu$ has no incoming arcs, we annotate it with $\infty$. 
    
    Let $V$ and $A$ be the sets of $\bigO(n \ell \log n)$ vertices and arcs of $G(S[i, i+\ell], T, Z)$.
    For the above annotation procedure, we first compute the $yx$-lexicographical ordering of the vertices, based on their corresponding points in the parameter space.
    This takes $\bigO(|V| \log |V|)$ time.
    Afterwards, we go over each vertex and each incoming arc exactly once, which take an additional $\bigO(|V| + |A|)$ time.
    In total, we annotate all vertices in $\bigO(n \ell \log^2 n)$ time.

\subparagraph*{Constructing the pathlets.}
    With the annotations, constructing the pathlets becomes straightforward.
    For each $j \in [\ell]$, we construct $\I_j$ as follows.
    We iterate over all critical point $(i+j, y')$ in the graph $G(S[i, i+\ell], T, Z)$. 
    For each critical point $(i+j, y')$ with a finite annotation $y$, we add the interval $[y, y']$ to $\I_j$.
    This procedure ensures that $\I_j$ contains all maximal intervals $[y, y']$ for which $\dF(S[i, i+j], T[y, y']) \leq \Delta'$, creating an optimal pathlet $(S[i, i+j], \I_j)$ with respect to its reference curve.
    Since there are $\bigO(n)$ critical points per $j$, this algorithm uses $\bigO(n \ell)$ time.
    Storing the pathlets takes $\bigO(n \ell)$ space.
    We conclude:

    \begin{theorem}
    \label{thm:constructing_vertex-to-vertex}
        Suppose $\Cov(C)$ is preprocessed by ~\cref{lem:coverage_ds}.
        In $\bigO(n^2 \ell \log^2 n)$ time and using $\bigO(n \ell \log n)$ space, we can construct an optimal vertex-to-vertex $(\ell, \Delta')$-pathlet $(P, \I)$.
    \end{theorem}
    \begin{proof}
        For a given vertex $S(i)$, we compute optimal pathlets $(S[i, i+j], \I_j)$ with respect to their reference curves for $j \in [\ell]$ in $\bigO(n \ell \log^2 n)$ time, using $\bigO(n \ell \log n)$ space.
        Using the data structure of~\cref{lem:coverage_ds}, we subsequently compute the coverage of one of these pathlets $\bigO(n \log n)$ time, so $\bigO(n \ell \log n)$ time for all.
        We pick the best pathlet and remember its coverage.
        Doing so for all vertices $S(i)$ of $S$, we obtain $|S|$ pathlets, of which we report the best.
        By only keeping the best pathlet in memory, rather than all $|S|$, the space used by these pathlets is lowered from $\bigO(n^2)$ to $\bigO(n)$.
    \end{proof}

\section{Subedge pathlets}
\label{sec:subedge_pathlets}

    Let $\Delta' = 4\Delta$, and let $C$ be a set of pathlets. We assume that $\Cov(C)$ has at most $\bigO(n^2 \log n)$ connected components. 
    We provide an algorithm for constructing a subedge $(2, \Delta')$-pathlet $(P, \I)$ whose coverage -- (the sum of lengths in $\Cov(P, \I) \backslash \Cov(C)$)  -- is at least one-eighth the optimum.

    Recall that a subedge pathlet $(P, \I)$ is a pathlet where $P = e[x, x']$ is a subsegment of some edge $e$ of $S$.
    We construct a subedge pathlet given the edge $e$.
    We first discretize the problem, identifying a set of $\bigO(n^2)$ \emph{critical points} in $\FSD(e, T)$.
    This set ensures that there exists a subedge pathlet $(e[x, x'], \I)$ with at least one-fourth the coverage of any pathlet using a subedge of $e$ as a reference curve, where for all $[y, y'] \in \I$, the points $(x, y)$ and $(x', y')$ are both critical points.

    For $j \in [n-1]$, consider the following six extreme points of $\FSD(e, T) \cap ([1, 2] \times [j, j+1])$ (where some points may not exist):
    \begin{itemize}
        \item A leftmost point of $\FSD(e, T) \cap ([1, 2] \times [j, j+1])$,
        \item A rightmost point of $\FSD(e, T) \cap ([1, 2] \times [j, j+1])$,
        \item The leftmost and rightmost points of $\FSD(e, T) \cap ([1, 2] \times \{j\})$, and
        \item The leftmost and rightmost points of $\FSD(e, T) \cap ([1, 2] \times \{j+1\})$.
    \end{itemize}
    Let $X_j$ be the set of $x$-coordinates of these points, and let $X = \bigcup X_j$ be the set of all these coordinates.
    Let $x_1, \dots, x_m$ be the set of values in $X$, sorted in increasing order.
    We call every point $(x_i, y)$ that is an endpoint of a connected component (vertical segment) of $\FSD(e, T) \cap (\{x_i\} \times [1, n])$ a critical point.
    Let $Z$ be the set of at most $4n^2 = \bigO(n^2)$ critical points.

    Before we restrict pathlets to be defined by these critical points, we first allow a broader range of pathlets.
    We consider the edge $\rev{e}$, obtained by reversing the direction of $e$, and look at constructing a pathlet that is a subedge of either $e$ or $\rev{e}$.
    We show that by restricting pathlets to be defined by $Z$, while allowing for reference curves that are subcurves of $\rev{e}$, results in losing only a factor four in the maximum coverage.
    
    \begin{figure}
        \centering
        \includegraphics{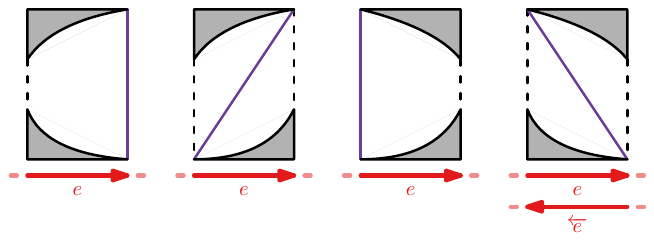}
        \caption{The connected components of $\FSD(e', T)$ fall into these four cases, based on where the minima and maxima of the bottom and top parabolic arcs lie.
        In the first three cases, there is a clear matching with optimal coverage (purple).
        In the fourth case, the matching is only valid when mirroring the free space, achieved by using $\rev{e}$ instead of $e$.}
        \label{fig:subedge_pathlets}
    \end{figure}

    \begin{restatable}{lemma}{subedgecriticalpoints}
     \label{lem:subedge_critical_points}
         Let $C$ be a set of pathlets.
        For any subedge $(2, \Delta')$-pathlet $(e[x, x'], \I)$, there exists a subedge $(2, \Delta')$-pathlet $(P, \I')$ with
        \[
            | \Cov(P, \I') \setminus \Cov(C) | \geq \frac{1}{4} | \Cov(e[x, x'], \I) \setminus \Cov(C) |,
        \]
        where $P$ is equal to $e[x_i, x_j]$ or $\rev{e}[x_i, x_j]$ for some $i$ and $j$, and for every interval $[y, y'] \in \I'$, the points $(x_i, y)$ and $(x_j, y')$ are contained in $Z$.        
    \end{restatable}
    \begin{proof}
        Consider a subedge $(2, \Delta')$-pathlet $(e[x, x'], \I)$.
        Any interval $[a, b] \in \I$ corresponds to a bimonotone path from $(x, a)$ to $(x', b)$ in $\FSD(e, T)$.
        Consider such an interval $[a, b]$ and a corresponding path $\pi$.
        
        Suppose first that $x_i \leq x \leq x' \leq x_{i+1}$ for some $i$.
        Observe that every connected component of $\FSD(e, T) \cap ([x_i, x_{i+1}] \times [1, n])$ is bounded on the left and right by (possibly empty) vertical line segments, and that the bottom and top chains are parabolic curves whose extrema are the endpoints of these segments.
        In particular, these connected components are convex.
        Thus there is a straight line segment $e'$ from $(x, a)$ to $(x', b)$ in the free space.
        The line segment $e^*$ connecting the extrema of the parabolic curves bounding the connected component containing $(x, a)$ and $(x', b)$ is longer than $e'$.
        The endpoints of $e^*$ are both critical points, and $e^*$ describes a valid matching between a subcurve of $T$ and either a subcurve of $e$, or a subcurve of $e'$.
        See~\cref{fig:subedge_pathlets}.
        As there are four different reference curves we choose from, the resulting intervals are spread over four different pathlets.
        Therefore, one of the pathlets must have at least one-fourth the coverage of any subedge pathlet.
        
        Next suppose that $x_i \leq x \leq x_{i+1} < x'$ for some $i$.
        At some point, $\pi$ reaches a point $(x_{i+1}, a')$.
        Let $(x^*, y^*)$ be the lowest point in the connected component containing $(x, a)$.
        This point is a critical point.
        By convexity, the segment $e^*$ from $(x^*, y^*)$ to $(x_{i+1}, a')$ lies in the free space.
        Because $y^* \leq a$ by our choice of $(x^*, y^*)$, we may connect $e^*$ to the suffix of $\pi$ that starts at $(x_{i+1}, a')$ to obtain a matching that starts at a critical point and that covers at least as much of $T$ as the original matching.
        Applying a symmetric procedure to the end $(x', b)$ of $\pi$ yields a matching that starts and ends at critical points without losing coverage.
        Again, since there are four different reference curves to choose from for the intervals in $\I$, the resulting intervals are spread over four different pathlets.
        One of the pathlets must therefore have at least one-fourth the coverage of any subedge pathlet.
    \end{proof}

    We find for each point $e(x_i)$ of $e$ corresponding to a critical point a subedge pathlet whose reference curve starts at $e(x_i)$ and ends at some point $e(x_j)$ that also corresponds to a critical point.
    To this end, we consider each point $e(x_i)$ separately.
    We proceed akin to the construction for vertex-to-vertex pathlets~\cref{sec:vertex-to-vertex_pathlets}, with some optimization steps.
    
    It proves too costly to consider each reference curve $e[x_i, x_{i'}]$ for every $x_i$ we consider.
    By sacrificing the quality of the pathlet slightly, settling for a pathlet with at least one-eighth the coverage of any subedge pathlet rather than one-fourth, we can reduce the number of reference curves we have to consider from $\Theta(m^2) = \bigO(n^2)$ to $\bigO(m \log m)$.
    Let $(e[x_i, x_{i'}], \I)$ be a subedge pathlet.
    We can split $e[x_i, x_{i'}]$ into two subedges $e[x_i, x_{i+2^j}]$ and $e[x_{i'-2^j}, x_{i'}]$.
    The matchings corresponding to $\I$ naturally decompose into two sets of matchings (whose matched subcurves may overlap), giving rise to two pathlets $(e[x_i, x_{i+2^j}], \I_1)$ and $(e[x_{i'-2^k}, x_{i'}], \I_2)$ with $\I_1 \cup \I_2 = \I$.
    Thus at least one of these pathlets has at least half the coverage that $(e[x_i, x_{i'}], \I)$~has.
    By~\cref{lem:subedge_critical_points}, a pathlet $(e[x_i, x_{i+2^j}], \I)$ that has maximum coverage out of all such pathlets then covers at least one-eighth of what any other subedge pathlet $(e[x, x'], \I')$ covers.

    We create a sweepline algorithm that, for each $e[x_i, x_{i+2^j}]$ (with $j \leq \log (m-i)$), constructs a reference-optimal $(\ell, \Delta')$-pathlet $(e[x_i, x_{i+2^j}], \I_j)$.
    We let each interval $\I_j$ contain all maximal intervals $[y, y']$ for which $\dF(e[x_i, x_{i+2^j}], T[y, y']) \leq \Delta'$, and thus all maximal intervals for which $(x_i, y)$ can reach $(x_{i+2^j}, y')$ by a bimonotone path in $\FSD(S, T)$.
    Note that both $(x_i, y)$ and $(x_{i+2^j}, y')$ are critical points.
    Thus we aim to find all maximal intervals $[y, y']$ for which the critical point $(x_i, y)$ can reach the critical point $(x_{i+2^j}, y')$ by a bimonotone path in $\FSD(S, T)$.

    Let $Z_i$ be the subset of $\bigO(n \log n)$ critical points with $x$-coordinate equal to $x_i$ or $x_{i+2^j}$ for some $j \leq \log (m-i)$.
    We construct, for each $i \in [n]$, the reachability graph $G(e, T, Z_i)$ from~\cref{sec:reachability_graph}, which encodes reachability between the critical points in $Z_i$.
    This graph takes $\bigO((n + |Z_i|) \log (n|Z_i|)) = \bigO(n \log^2 n)$ time to construct and has complexity $\bigO(n \log^2 n)$ (see~\cref{thm:reachability_graph}).
    We aim to annotate each vertex $\mu$ (note that $\mu$ does not have to be a critical point) in $G(e, T, Z_i)$ with the minimum $y$, such that there exists a critical point $(x_i, y)$ that can reach $\mu$.
    We annotate $\mu$ with $\infty$ if no such $y$ exists.
    
\subparagraph*{Annotating vertices and asymptotic analysis.}
    We first annotate the vertices $(x_i, y)$ in $\bigO(n)$ time by scanning over them in order of increasing $y$-coordinate.
    We process the remaining vertices in $yx$-lexicographical order, first by increasing $y$-coordinate, and  by increasing $x$-coordinate when ties arise.
    Each vertex $\mu$ that we consider has only incoming arcs that originate from vertices below and left of $\mu$.
    By our lexicographical ordering, each of these vertices are already annotated. 
    The minimal $y$ for which there exists a path from $(x_i, y)$ to $\mu$, must be the minimum over all its incoming arcs which we compute in time proportional to the in-degree of $\mu$. 
    If $\mu$ has no incoming arcs, we annotate it with $\infty$. 
    
    Let $V$ and $A$ be the sets of $\bigO(n \log^2 n)$ vertices and arcs of $G(e, T, Z_i)$.
    For the above annotation procedure, we first compute the $yx$-lexicographical ordering of the vertices, based on their corresponding points in the parameter space.
    This takes $\bigO(|V| \log |V|)$ time.
    Afterwards, we go over each vertex and each incoming arc exactly once, which take an additional $\bigO(|V| + |A|)$ time.
    In total, we annotate all vertices in $\bigO(n \log^3 n)$ time.

\subparagraph*{Constructing the pathlets.}
    Using the annotations, constructing the pathlets becomes straightforward.
    For each $j \in [\log (m-i)]$, we construct $\I_j$ as follows.
    We iterate over all critical point $(x_{i+2^j}, y')$ in the graph $G(e, T, Z_i)$. 
    For each critical point $(x_{i+2^j}, y')$ with a finite annotation $y$, we add the interval $[y, y']$ to $\I_j$.
    This procedure ensures that $\I_j$ contains all maximal intervals $[y, y']$ for which $\dF(e[x_i, x_{i+2^j}], T[y, y']) \leq \Delta'$, making an optimal pathlet $(e[x_i, x_{i+2^j}], \I_j)$ with respect to its reference curve.
    As there are $\bigO(n)$ critical points per $j$, this algorithm uses $\bigO(n \log n)$ time.
    Storing the pathlets takes $\bigO(n \log n)$ space.
    Thus, we conclude the following:

    \begin{lemma}
        Let $C$ be a set of pathlets where $\Cov(C)$ has $\bigO(n^2 \log n)$ connected components.
        Suppose $\Cov(C)$ is preprocessed into the data structure of~\cref{lem:coverage_ds}.
        In $\bigO(n^2 \log^3 n)$ time and using $\bigO(n \log^2 n)$ space, we can construct a $(2, \Delta')$-pathlet $(P, \I)$ with
        \[
            \lVert \Cov(P, \I) \setminus \Cov(C) \rVert \geq \frac{1}{8} \lVert \Cov(P', \I') \setminus \Cov(C) \rVert
        \]
        for any $(2, \Delta')$-pathlet $(P', \I')$ where $P'$ is a subsegment of a given directed line segment $e$.
        The intervals in $\I$ all have endpoints that come from a set of at most $4n^2$ values.
    \end{lemma}
    \begin{proof}
        For a given point $e(x_i)$, we compute optimal pathlets $(e[x_i, x_{i+2^j}], \I_j)$ with respect to their reference curves for $j \in [\log (m-i)]$ in $\bigO(n \log^3 n)$ time, using $\bigO(n \log n)$ space.
        Using the data structure of~\cref{lem:coverage_ds}, we subsequently compute the coverage of one of these pathlets $\bigO(n \log n)$ time, so $\bigO(n \log^2 n)$ time for all.
        We pick the best pathlet and remember its coverage.
        Doing so for all points $e(x_i)$, we obtain $m$ pathlets, of which we report the best.
        This pathlet has at least one-eighth the coverage of any other subedge $(2, \Delta')$-pathlet $(e[x, x'], \I)$.
        By only keeping the best pathlet in memory, rather than all $m$, the space used by these pathlets is lowered from $\bigO(mn)$ to $\bigO(n)$.
    \end{proof}
    
    \begin{theorem}
    \label{thm:constructing_subedge}
        Suppose that the universe $\U$ and the coverage $\Cov(C)$ is preprocessed into the data structure of~\cref{lem:coverage_ds}.
        In $\bigO(n^3 \log^3 n)$ time and using $\bigO(n \log^2 n)$ space, we can construct a $(2, \Delta')$-pathlet $(P, \I)$ with
        \[
            | \Cov(P, \I) \setminus \Cov(C) | \geq \frac{1}{8} | \Cov(P', \I') \setminus \Cov(C) |
        \]
        for any subedge $(2, \Delta')$-pathlet $(P', \I')$.
    \end{theorem}

\section{Conclusion}
    In this work, we presented an improved approximation algorithm for subtrajectory clustering.
    We discuss our technical contribution, and how it differs from previous works, our asymptotic improvements and finally interesting directions for future work. 

    \subparagraph{Technical contribution.}   
    Our technical contributions are threefold: 

    First, we introduced a new type of curve simplification in Section~\ref{sec:simplification}.
    This simplification allows us to construct a curve $S$, such that our clustering needs to consider only pathlets whose reference curve is a subcurve of $S$.
    Although numerous similar curve-simplification algorithms exist, our method distinguishes itself by lying significantly closer to the input curve $T$. Consequently, our approximation algorithm is a $4$-approximation in $\Delta$, compared to the $11$-approximations of prior works. We consider this simplification to be of independent interest, as future works may immediately use our simplification method to obtain $4$-approximations in $\Delta$ also. 

    Secondly, we considered in Section~\ref{sec:the_algorithm} an extension to the greedy set cover algorithm wherein each iteration adds an approximately-maximum covering element, rather than a maximum one.
    Observe that $P$ can always be divided into at most three subcurves, where at most one of them starts and ends at a vertex of $S$ (a vertex-subcurve) and at most two of them are subcurves of an edge of $S$ (a subedge of $S$).
    We design a greedy meta-algorithm, that in each iteration computes an $(\ell, \Delta)$-pathlet $(P, \mathcal{I})$ with approximately-maximum coverage, whose reference curve is a vertex-subcurve or subedge of $S$.   
    Our approximately greedy set cover analysis shows that our meta-algorithm computes a clustering of size $\bigO(k \log n)$. A key takeaway from our construction is that by restricting our attention to vertex-subcurves and subedges of $S$, we significantly reduce the set of candidate pathlets from $\Ot(n^3 \ell)$ to $\Ot(n^2)$.
    We consider this fact to also be of independent interest. Indeed, our subsequent algorithm spends near-linear time per candidate pathlet but future works may discover more efficient algorithms over the same smaller candidate set. 

    Finally, we presented algorithms in Sections~\ref{sec:vertex-to-vertex_pathlets} and \ref{sec:subedge_pathlets} that compute the corresponding candidate pathlet for a candidate reference curve in near-linear time and near-linear space. 
    The key observation to this contribution is that we show that it suffices to compute all candidate pathlets on the fly, significantly reducing the space.

    \subparagraph{Asymptotic improvements.}
    Compared to the best prior deterministic work~\cite{conradi2023finding}, our algorithm improves the running time by a factor near-linear in $n \ell$, improves the space by a factor near-linear in $n^2 \ell$, and improves the approximation in $\Delta$ from a factor $11$ to $4$, all whilst asymptotically matching the size of the clustering. 
    We consider this a significant improvement over the state-of-the-art.
    
    When we compare to previous randomized work~\cite{bruning_faster_2022, bruning_subtrajectory_2023} 
    we improve the running time by a factor near-linear in $\ell$,
     improve space by a factor $n$, and improve the approximation in $\Delta$ from a factor $11$ to $4$.
    A downside of our approach is that, compared to randomised works, we only asymptotically match the clustering size whenever $\ell$ is relatively large  (i.e., $\ell \in \Omega(\log n / \log k)$).
    However, we note that on all other algorithmic quality measures, we still offer a substantial improvement whilst also being deterministic. In addition, when considering algorithmic performance in practice, we note that these previous randomized results~\cite{bruning_faster_2022, bruning_subtrajectory_2023}  use $\eps$-net sampling.
    Such a sampling procedures leads to very high hidden constants in the asymptotic clustering size which makes such an approach impractical.     

    \subparagraph{Future work.}
    We think it remains an interesting open problem whether one can obtain a clustering size of $\bigO(k \ell \log k)$ in a deterministic manner. 
    We also note that, currently, our algorithm considers a set of $\Ot(n^2)$ reference curves $P$, and computes an $(\ell, \Delta)$-pathlet $(P, \I)$ with approximately-maximum coverage for each reference curve independently, in near-linear time. 
    We think it is an interesting open problem whether one can present an algorithm that is overall more efficient whenever these maximum pathlets are considered simultaneously rather than independently. 
    
\bibliography{bibliography}

\appendix

\section{The interior-disjoint setting}

    Previous definitions of subtrajectory clustering have imposed various restrictions on the pathlets in the clustering. 
    For example, in~\cite{buchin_detecting_2011, buchin2017clustering,buchin_improved_2020, gudmundsson2022cubic} the pathlets must be \emph{interior-disjoint}.
    A pathlet $(P, \I)$ is interior-disjoint whenever the intervals in $\I$ are pairwise interior-disjoint.
    While we do not give dedicated algorithms for the interior-disjoint setting, we show in~\cref{lem:constructing_interior-disjoint} that we can efficiently convert any pathlet into two interior-disjoint pathlets with the same coverage.
    This gives a post-processing algorithm for converting a clustering $C$ into an interior-disjoint clustering $C'$ with at most twice the number of pathlets.
    We first show the following auxiliary lemma.

    \begin{lemma}
    \label{lem:reducing_ply}
        Given a set of intervals $\I$, we can compute a subset $\I' \subseteq \I$ with ply\footnote{
            The ply of a set of intervals is the maximum number of intervals with a common intersection.
        } at most two and with $\bigcup \I' = \bigcup \I$ in $\bigO(|\I| \log |\I|)$ time.
    \end{lemma}
    \begin{proof}
        We first sort the intervals of $\I$ based on increasing lower bound.
        We then remove all intervals in $\I$ that are contained in some other interval in $\I$, which can be done in a single scan over $\I$ by keeping track of the largest endpoint of an interval encountered so far.
        We initially set $\hat{\I} = \emptyset$ and iterate over the remaining intervals in order of increasing lower bound.
        During iteration, we maintain the invariant that $\hat{\I}$ has ply at most two.
        Let $I_1, \dots, I_k$ be the intervals in $\hat{\I}$ in order of increasing lower bound.
        Suppose we consider adding an interval $I \in \I$ to $\hat{\I}$.
        If $I \subseteq \hat{\I}$, then we ignore $I$, since it does not add anything to the coverage of $(P, \hat{I})$.
        Otherwise, we set $\hat{\I} \gets \hat{\I} \cup \{I\}$.
        This may have increased the ply of $\hat{\I}$ to three, however.
        We next show that in this case, we can remove an interval from $\hat{\I}$ to decrease the ply back to two, without altering $\bigcup \hat{\I}$.
        
        Observe that if the ply of $\hat{\I}$ increases to three, then $I_{k-1}$, $I_k$ and $I$ must intersect.
        Indeed, $I$ must have a common intersection with two other intervals in $\hat{\I}$.
        Suppose for sake of contradiction that there is some $I_i \in \hat{\I}$ that intersects $I_i$ for some $i < k-1$.
        Then $I_i$ must contain the lower bounds of $I_{k-1}$ and $I_k$.
        However, $I_{k-1}$ must then also contain the lower bound of $I_k$, as otherwise $I_{k-1} \subset I_i$, which means that $I_{k-1}$ was already filtered out at the beginning of the algorithm.
        Thus, $I_i$, $I_{k-1}$ and $I_k$ have a common intersection (the lower bound of $I_k$), which contradicts our invariant that $\hat{\I}$ has ply at most two.
        Now that we know that $I_{k-1}$, $I_k$ and $I$ intersect, note that $I_k \subseteq I_{k-1} \cup I$, since the lower bound of $I_k$ lies between those of $I_{k-1}$ and $I$, and $I \nsubseteq I_k$, so the upper bound of $I_k$ lies between those of $I_{k-1}$ and $I$ as well.
        Hence we can set $\hat{\I} \gets \hat{\I} \setminus \{I_k\}$ to reduce the ply back to two, while keeping $\bigcup \hat{\I}$ the same.
        After sorting $\I$, the above algorithm constructs $\hat{\I}$ in $\bigO(|\I|)$ time.
        This gives a total running time of $\bigO(|\I| \log |\I|)$.
    \end{proof}

    \begin{lemma}
    \label{lem:constructing_interior-disjoint}
        Given an $(\ell, \Delta)$-pathlet $(P, \I)$, we can construct two interior-disjoint $(\ell, \Delta)$-pathlets $(P_1, \I_1)$ and $(P_2, \I_2)$ with $\I_1 \cup \I_2 = \I$ in $\bigO(|\I| \log |\I|)$ time.
    \end{lemma}
    \begin{proof}
        First construct a subset $\I' \subseteq \I$ with ply at most two and $\bigcup \I' = \bigcup \I$ using~\cref{lem:reducing_ply}.
        Then sort $\I'$ based on increasing lower bound.
        Construct $\I_1$ by iterating over $\I'$ and greedily taking any interval that is interior-disjoint from the already picked intervals.
        Finally, set $\I_2 \gets \I' \setminus \I_1$.
    \end{proof}

\section{Constructing a pathlet-preserving simplification}
\label{app:constructing_pathlet_preserving_simplification}

\subsection{Defining our \texorpdfstring{$2\Delta$}{2Δ}-simplification \texorpdfstring{$S$}{S}}
\label{appsub:greedy_simplification}
        
    We consider the vertex-restricted simplification defined by Agarwal~\etal~\cite{agarwal_near-linear_2005} and generalize their $2\Delta$-simplification definition, allowing vertices to lie anywhere on $T$ (whilst still appearing in order along $T$).
    This way, we obtain a simplification with at most as many vertices as the optimal unrestricted $\Delta$-simplification (see \cref{fig:simplification_components}).
        
    \begin{figure}
        \centering
        \includegraphics{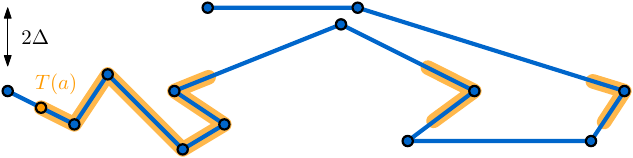}
        \caption{Consider the trajectory $T$ and some point $T(a)$. For some $\Delta$, we can indicate all $T(b) \in T$ with $a \leq b$ where for line $s = \overline{T(a)T(b)}$, $\dF(s, T[a, b]) \leq 2 \Delta$. Note that this set $\B(a)$ is not a connected set of intervals on $[1, n]$.}
        \label{fig:simplification_components}
    \end{figure}
        
    \begin{definition}
    \label{def:baset}
        Let $T$ be a trajectory with $n$ vertices, $\Delta \geq 0$ and $a \in [1, n]$. 
        We define the set $\B(a) = \{ b \geq a \mid \dF(\overline{T(a) T(b)}, T[a, b]) \leq 2\Delta \}$. 
    \end{definition}
        
    \begin{definition}
    \label{def:curve_S}
        Let $T$ be a trajectory with $n$ vertices and $\Delta \geq 0$.
        We define our $2\Delta$-simplified curve $S$ as follows: the first vertex of $S$ is $T(1)$. 
        The second vertex of $S$ may be any point $T(a)$ with $a$ as a rightmost endpoint of an interval in $\B(1)$.
        The third vertex of $S$ may be any point $T(b)$ with $b$ as a rightmost endpoint of an interval in $\B(a)$, and so forth. 
    \end{definition}
 
    Per definition of the set $\B(a)$, the resulting curve $S$ is a $2\Delta$-simplified curve.
    Let $(f, g)$ be any $2\Delta$-matching between $S$ and $T$ that matches the vertices of $S$ to the points on $T$ that define them.
    We prove that $(S, f, g)$ is a pathlet-preserving simplification.
    
    \begin{lemma}
    \label{lem:pathlet_preserving}
     The curve $S$ as defined above, together with the matching $(f, g)$, forms a pathlet-preserving simplification. 
    \end{lemma}
    \begin{proof}
        We show that for any subcurve $T[a, b]$ and all curves $P$ with $\dF(P, T[a, b]) \leq \Delta$, the subcurve $S[s, t]$ matched to $T[a, b]$ by $(f, g)$ has complexity $|S[s, t]| \leq |P| + 2 - |\N \cap \{s, t\}|$.
        For brevity, we write $X = T[a, b]$.
        
        Fix any curve $P$ with $\dF(P, X) \leq \Delta$. There exists a $\Delta$-matchings $(f', g')$ between $P$ and $X$.
        Per construction, the subcurve $S[s, t]$ has \f distance $\dF(S[s, t], X) \leq 2\Delta$ to $X$.
        Any vertex of $T$ that is a vertex of $S[s, t]$ is also a vertex of $T[a, b]$.
        Let $S[x, y]$ be the maximal vertex subcurve of $S[s, t]$.
        This curve naturally has $|S[s, t]| - 2 + |\N \cap \{s, t\}|$ vertices.
        We argue that $|P| \geq |S[x, y]|$.

        Suppose for sake of contradiction that $|P| < |S[x, y]|$.
        By the pigeonhole principle, there must exist an edge $e_S = \overline{T(\alpha) T(\beta)}$ of $S[x, y]$, as well as an edge $e_P$ of $P$ matched to some subcurve $T[c, d]$ of $T$ by $(f', g')$, such that $c \leq \alpha \leq \beta < d$.
        We claim that $\beta' \in \B(\alpha)$ for all $\beta' \leq d$.

        The proof is illustrated in~\cref{fig:closeproof}.
        For any $\beta' \leq d$ there exists a subedge $e = e_P[x_1, x_2]$ of $e_P$ that is matched to $T[\alpha, \beta']$ by $(f, g)$.
        Per definition of a $\Delta$-matching we have that $\dF(T[\alpha, \beta'], e') \leq \Delta$.
        The $2\Delta$-matching implies that $\lVert e(1) - T(\alpha) \rVert \leq 2\Delta$ and $\lVert e(2) - T(\beta') \rVert \leq 2\Delta$.
        We use this fact to apply~\cite[Lemma~3.1]{agarwal_near-linear_2005} and note that $\dF(e, \overline{T(\alpha)T(\beta')}) \leq \Delta$.
        Applying the triangle inequality, we conclude that $\dF(T[\alpha, \beta'], \overline{T(a)T(b')}) \leq 2 \Delta$. 
        It follows that $\beta' \in \B(\alpha)$. 
        
        We obtain that $[\alpha, d]$ is contained in the first connected component of $\B(\alpha)$. 
        However, per construction of $S$, $\beta$ is a rightmost endpoint of a connected component in $\B(\beta)$.
        This contradicts the fact that $\beta \in [\alpha, d)$.
    \end{proof}
    
    \begin{figure}[h]
        \centering
        \includegraphics{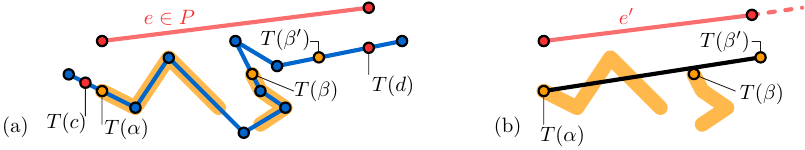}
        \caption{(a) The construction in the proof of Lemma~\ref{lem:pathlet_preserving}.
        We have an edge $e$ with $\dF(e, T[c, d]) \leq \Delta$. Moreover, for some $\alpha \in [c, d]$ we show $\B(\alpha)$ in orange where $\beta$ is the last value in some connected component of $\B(\alpha)$. (b) For any $\beta' \in [\alpha, d]$, there exists a subedge $e'$ of $e$ with $\dF(e', T[a, \beta']) \leq \Delta$.  }
        \label{fig:closeproof}
    \end{figure}
        
\subsection{Constructing the simplification}
        
    We give an $\bigO(n \log n)$ time algorithm for constructing our greedy simplification $S$ (\cref{def:curve_S}).
    Given any point $T(a)$ on $T$, we decompose the problem of finding a $b \in \B(a)$ over all edges of $T$.
    That is, given an $a \in [1, n]$, we consider an individual edge $T[i, i+1]$ of $T$. 
    We show how to report the maximal $b \in [i, i+1] \cap \B(a)$.
    Our procedure is based off of the work by Guibas~\etal~\cite{guibas93minimum_link} on ordered stabbing of disks in $\mathbb{R}^2$, and takes $\bigO(i-a)$ time.
    We fix a plane $H$ in $\R^d$ that contains $T(a)$ and $T[i, i+1]$. 
    On a high-level, we apply the argument by Guibas~\etal in $\R^d$ by restricting the disks to their intersection with $H$: 
    
    \begin{definition}
        Let $H$ be some fixed two-dimensional plane in $\R^d$. 
        For any $x \in [1, n]$ denote by $B_x$ be the ball in $H$ that is obtained by intersecting a ball of radius $2\Delta$ centered at $T(x)$ with $H$.
        For $a \leq b$ we say that a directed line segment $e$ in $H$ \emph{stabs} all balls in $[a, b]$ in order if for all $k \in \{a\} \cup ([a, b] \cap \N) \cup \{b\}$ there are points $p_k \in e \cap B_k$ such that $p_k$ comes before $p_{k'}$ on $e$ whenever $k \leq k'$ (see~\cref{fig:coneconstruction} (a)). 
    \end{definition}
    
    \begin{lemma}[{\cite[Theorem~14]{guibas93minimum_link}}]
        \label{lem:categorization}
        A line segment $e$ is within \f distance $2\Delta$ of a subcurve $T[a, b]$ of $T$ if and only if the following conditions are met:
        \begin{enumerate}
            \item $e$ starts within distance $2\Delta$ of $T(a)$,
            \item $e$ ends within distance $2\Delta$ of $T(b)$, and
            \item $e$ stabs all balls in $[a, b]$ in order. 
        \end{enumerate}
    \end{lemma}        
        
\subparagraph*{Computing the maximal $b \in [i, i+1] \cap \B(a)$.}
    For any edge $e = \overline{T(a)T(b)}$ of $Z$, the endpoints lie on $T$ and thus $e$ trivially satisfies the first two criteria. 
    It follows that if we fix some $T(a)$ on $S$ and some edge $T[i, i+1]$, then the maximal $b \in [i, i+1]$ (with $b \geq a$) for which $\overline{T(a) T(b)}$ stabs balls $[a, b]$ in order is also the maximal $b \in [i, i+1] \cap \B(a)$. 
    We consider the following (slightly reformulated) lemma by Guibas~\etal~\cite{guibas93minimum_link}:
    
    \begin{lemma}[{\cite[Lemma~8]{guibas93minimum_link}}]
        \label{lem:ordered_intervals}
        Let $[a_j, b_j]$ be a sequence of intervals. 
        There exist $p_j \in [a_j, b_j]$ with $p_j \leq p_k$ for all $j \leq k$, if and only if there is no pair $j \leq k$ with $b_k < a_j$.
    \end{lemma}
    
    The above lemma is applicable to segments in $H$ stabbing balls in $H$. 
    Indeed, consider all integers $j \in [a, i]$. 
    We may view any directed line segment $e$ in $H$ as (part of) the real number line, and view the intersections between $e$ and the disks $D_j$ as intervals.
    \Cref{lem:ordered_intervals} then implies that $e$ stabs $[\ceil{a}, i]$ in order, if and only if no integers $j, k \in [a, i]$ exist with $j \leq k$ such that $e$ leaves $D_k$ before it enters $D_j$ (assuming $e$ intersects all disks).
    
    For all integers $j \in [a, i]$, let $W_j := \{ p \in H \mid \overline{T(a) p} \textnormal{ intersects } D_j \}$. 
    We define the \emph{stabbing wedge} $\SW_j := \{ p \in H \mid \overline{T(a) p} \textnormal{ intersects } [a, j]  \textnormal{ in order} \}$. 
    We prove the following:
    
    \begin{lemma}
        \label{lem:intersection}
        Either $\SW_j = \bigcap_{ k \in [a, j] \cap \N } W_k$, or $\SW_j = \emptyset$.
    \end{lemma}
    \begin{proof}
        The proof is by induction. 
        The base case is that, trivially, $\SW_{\ceil{a}} = W_{\ceil{a}}$.
        Any line $\overline{T(a) p}$ that intersects $[a, j]$ also intersects $[a, j-1]$, thus whenever $\SW_{j-1}$ is empty, then $\SW_j$ must also be empty. 
        Suppose now that $\SW_{j-1} = \bigcap_{ k \in [a, j - 1] \cap \N } W_k$. We show that  $\SW_j$ is either equal to $\SW_{j-1} \cap W_j$, or it is empty (which, by induction, shows the lemma). 
         
        First we show that $\SW_j \subseteq \SW_{j-1} \cap W_j$.
        Suppose $\SW_j$ is non-empty, and take a point $p \in \SW_j$.
        By definition of stabbing wedge $\SW_j$, the segment $\overline{T(a) p}$ stabs $[a, j]$ in order, and thus $[a, j-1]$ in order. 
        Furthermore, $p$ must lie in $W_j$ for $\overline{T(a) p}$ to be able to stab disk $D_j$.
    
        Next we show that $ \SW_{j-1} \cap W_j \subseteq \SW_j$.
        By~\cref{lem:ordered_intervals}, $p \in \SW_j$ if and only if $\overline{T(a) p}$ first enters all disks $D_{\ceil{a}}, \dots, D_{j-1}$ before exiting disk $D_j$.
        Fix some $p \in \SW_{j-1} \cap W_j$, for all $k < j$ the line $\overline{T(a) p}$ intersects $D_k$. 
        If $p \neq \in \SW_j$ then there must exist a $k < j$ where $\overline{T(a) p}$ exists $D_j$ before it enters $D_k$. 
        The area $\SW_{j-1}$ must be contained in the cone $C \subset H$ given by $T(a)$ and the two tangents of $D_k$ to $T(a)$ (\cref{fig:coneconstruction} (b)) and thus $\overline{T(a) p}$ is contained in $C$. 
        Since $\overline{T(a) p}$ intersects $D_k$ in $C$ after $D_j$, it must be that $D_j \cap C$ is contained in a triangle $C^*$ formed by the boundary of $C$ and another tangent of $D_k$. 
        However, this means that any segment $\overline{T(a) q}$ (for $q \in \SW_{j-1} \cap W_j)$ that stabs the disks $[a, j - 1]$ in order must intersect $C^*$ before it intersects $D_k$. Thus, by Lemma~\ref{lem:ordered_intervals}, there is no segment  $\overline{T(a) q}$ that stab the disks $[a, j]$ in order and $\SW_j$ is empty. 
        Thus, either $\SW_j$ is empty or $\SW_j = \bigcap_{\ceil{a} \leq k \leq j} W_k$. 
    \end{proof}
    
    \begin{figure}
        \centering
        \includegraphics{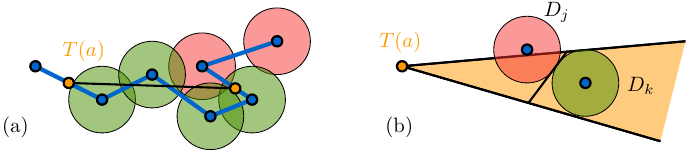}
        \caption{(a) For two points $T(a)$ and $T(b)$ on $T$, we consider all $j \in [a, b] \cap N$ and draw a ball with radius $2\Delta$ around them in green. 
        (b) For all $k \in [a, j-1] \cap \N$, the area $\SW_{j-1}$ must be contained in the orange cone. }
        \label{fig:coneconstruction}
    \end{figure}
    
    \begin{lemma}
        \label{lem:maximum_of_edge}
        Given $a \in [1, n]$ and edge $T[i, i+1]$ of $T$, we can compute the maximum $b \in [i, i+1] \cap \B(a)$, or report that no such $b$ exists, in $\bigO(1+i-a)$ time.
    \end{lemma}
    \begin{proof}
        By~\cref{lem:categorization} (and the fact that $T(a)$ and $T(b)$ always lie on $T$) $b \in [i, i+1] \cap \B(a)$ if and only if $T(b) \in \SW_i$.
        For any edge $T[i, i+1]$, we compute the maximal $b \in [i, i+1] \cap \B(a)$ by first assuming that $\SW_i$ is non-empty.
        We check afterwards whether this assumption was correct, and if not, we know that no $b \in [i, i+1]$ exists with $\dF(T[a, b], \overline{T(a) T(b)}) \leq 2\Delta$.
        
        For all $j \in [a, i] \cap N$, we compute the last value $b_j \in [i, j]$ such that $T(b_j) \in W_j$. 
        This can be done in $\bigO(1)$ time per integer $j$, as wedges in the plane are formed by two rays and a circular arc in $H$, which we can intersect in $\bigO(1)$ time. 
        Then we set $b = \min_j b_j$.
        
        To check whether $\SW_i$ is non-empty, we determine if $\dF(T[a, b], \overline{T(a) T(b)}) \leq 2\Delta$.
        This takes $\bigO(1+i-a)$ time with the algorithm of Alt and Godau~\cite{alt95continuous_frechet}.
        The assumption that $\SW_i$ is non-empty is correct precisely if $\dF(T[a, b], \overline{T(a) T(b)}) \leq 2\Delta$.
        If $\SW_i$ is empty, $[i, i+1] \cap \B(a)$ is empty and no output exists. 
        If $\SW_i$ is non-empty, then by Lemma~\ref{lem:intersection}, the $b$ we choose is the maximal $b \in [i, i+1] \cap \B(a)$. 
    \end{proof}
    
    \begin{lemma}
        \label{lem:computing_good_y}
        Given $a \in [1, n]$, we can compute a value $b^*$ that is the maximum of some connected component of $\B(a)$ in $\bigO((1+b^*-a) \log n)$ time.
    \end{lemma}
    \begin{proof}
        We use Lemma~\ref{lem:maximum_of_edge} in conjunction with exponential and binary search to compute the maximum $b^*$ of some connected component of $\B_{2\Delta}(a)$:
        
        We search over the edges $T[i, i+1]$ of $T$.
        For each considered edge we apply~\cref{lem:maximum_of_edge} which returns some $b \in [i, j+1] \cap \B(a)$ (if the set is non-empty). We consider three cases:
        
        If $b \in [i, i+1)$, then this value is the maximum of some connected component of $\B(a)$. We stop the search and output $b$.
        
        If the procedure reports the value $b = i+1$ then this value may not necessarily be the maximum of a connected component.  However, there is sure to be a maximum of at least $b$. Hence we continue the search among later edges of $T$ and discard all edges before, and including, $T[i, i+1]$.
        
        If the procedure reports no value then $[i, i+1] \cap \B(a) = \emptyset$.
        Since trivially $a \in \B(a)$, it must be that there is a connected component whose maximum is strictly smaller than $i$.
        We continue the search among earlier edges of $T$ and discard all edges after, and including, $T[i, i+1]$.
        
        The above algorithm returns the maximum $b^* \in [i^*, i^*+1)$ of some connected component of $\B(a)$.
        By applying exponential search first, the edges $T[i, i+1]$ considered all have $i \leq 2i^* - a$.
        Hence we compute $b^*$ in $\bigO((1+b^*-a) \log n)$ time.
    \end{proof}
    
    We now iteratively apply~\cref{lem:computing_good_y} to construct our curve $S$.
    We obtain a $2\Delta$-matching $(f, g)$ by constructing separate matchings between the edges $\overline{T(a) T(b)}$ of $S$ and the subcurves $T[a, b]$ that they simplify.
    By~\cref{lem:pathlet_preserving} this gives a pathlet-preserving simplification $(S, f, g)$.

    \thmSimplification*

    \end{document}